\documentclass[%
 reprint,
superscriptaddress,
 amsmath,amssymb,
 aps,
pra,
]{revtex4-1}

\usepackage{graphicx}
\usepackage{dcolumn}
\usepackage{bm}
\usepackage{amsthm}
\usepackage{hyperref} 

\usepackage{comment}
\usepackage{braket}
\usepackage{tikz}
\usetikzlibrary{calc}

\newtheorem{thm}{Theorem}

\newcommand{\Hcd}{H_\text{corr}}
\newcommand{\ucd}{U_\text{d}}
\newcommand{\ub}{U_\text{CD}}
\newcommand{\ua}{U_\text{E}}
\newcommand{\uad}{U_\text{ad}}
\newcommand{\hb}{H_\text{CD}}
\newcommand{\hcd}{H_\text{corr}}
\newcommand{\h}{\mathfrak{h}}
\newcommand{\p}{\mathfrak{p}}
\newcommand{\ha}{H_\text{E}}
\newcommand{\hh}{\bm{\mathrm{H}}}
\newcommand{\spn}{\text{span}\{\bm{\mathrm{H}}\}}
\newcommand{\spnn}{\text{span}\{-i\bm{\mathrm{H}}\}}

\newcommand{\norm}[1]{\left\lVert#1\right\rVert}
\newcommand{\hlz}{H_\text{LZM}}

\newcommand{\ma}{M_\text{E}}
\newcommand{\mb}{M_\text{CD}}
\newcommand{\mc}{M_\text{corr}}
\newcommand{\fb}{f_\text{CD}}
\newcommand{\su}{\mathfrak{su}}
\newcommand{\sx}{\sigma_x}
\newcommand{\sy}{\sigma_y}
\newcommand{\sz}{\sigma_z}
\DeclareMathOperator{\sign}{sgn}

\DeclareMathOperator{\tr}{Tr}
\newcommand{\If}{\mathcal{I}_\text{F}}

\begin{document}

\preprint{APS/123-QED}

\title{Fast adiabatic evolution by oscillating initial Hamiltonians}

\author{Francesco Petiziol}
 \email{francesco.petiziol@studenti.unipr.it}
\affiliation{Dipartimento di Scienze Matematiche, Fisiche e Informatiche, Universit\`a di Parma, Parco Area delle Scienze 7/a, 43124 Parma, Italy}
\affiliation{INFN, Sezione di Milano Bicocca, Gruppo Collegato di Parma, Parco Area delle Scienze 7/a, 43124 Parma, Italy}
\author{Benjamin Dive}%
\affiliation{Department of Physics, Imperial College, SW7 2AZ London, UK}%
\author{Florian Mintert}
\affiliation{Department of Physics, Imperial College, SW7 2AZ London, UK}
\author{Sandro Wimberger}
\affiliation{Dipartimento di Scienze Matematiche, Fisiche e Informatiche, Universit\`a di Parma, Parco Area delle Scienze 7/a, 43124 Parma, Italy}
\affiliation{INFN, Sezione di Milano Bicocca, Gruppo Collegato di Parma, Parco Area delle Scienze 7/a, 43124 Parma, Italy}

\begin{abstract}
We propose a method to produce fast transitionless dynamics for finite dimensional quantum systems without requiring additional Hamiltonian components not included in the initial control setup, remaining close to the true adiabatic path at all times. The strategy is based on the introduction of an effective counterdiabatic scheme: a correcting Hamiltonian is constructed which approximatively cancels nonadiabatic effects, inducing an evolution tracking the adiabatic states closely. This can be absorbed into the initial Hamiltonian by adding a fast oscillation in the control parameters. We show that a consistent speed-up can be achieved without requiring strong control Hamiltonians, using it both as a standalone shortcut-to-adiabaticity and as a weak correcting field. A number of examples are treated, dealing with quantum state transfer in avoided-crossing problems and entanglement creation.
\end{abstract}

\maketitle

\section{\label{sec:intro}Introduction}

Quantum adiabatic processes \cite{messiah1961qm} are ubiquitous in quantum science and they represent an important resource for quantum control.
The adiabatic theorem states that when the system Hamiltonian $H(t)$ is varied slowly enough in time from an initial configurations $H(t_i)$ to a final configuration $H(t_f)$, then a state initially in an eigenstate of $H(t)$ will remain so during the whole evolution \cite{messiah1961qm,kato1950}. This makes such protocols intrinsically robust against experimental imperfections. On the other hand, the necessity of very long timescales for their implementation dramatically limits the number of operations which can be performed on the system within reasonable coherence times.

Recently, much work has been done towards the design of adiabatic-inspired control strategies which, on the one hand, inherit the robustness of the adiabatic dynamics, while, on the other, avoid the necessity of slow driving.
Among such so-called ``shortcuts to adiabaticity'' (STAs), a particularly promising method was introduced under the name of counterdiabatic (CD) \cite{rice1} or transitionless \cite{berry1} quantum driving. The basic idea which is put forward is that it is always possible to reverse-engineer a correcting Hamiltonian $\hb(t)$ such that the  total Hamiltonian $H(t) + \hb(t)$ keeps the system in the instantaneous eigenvectors of $H(t)$ without requiring it to change slowly -- that is, such that the adiabatic dynamics is an exact solution of the time-dependent Schr\"odinger equation.

When facing a quantum control problem, one must take into account that only a fairly restricted number of Hamiltonians can be realized and controlled in practice. The crucial drawback of the CD method is that, although a well-defined expression for computing $\hb(t)$ exists, the correcting field typically requires time-dependent control of complex interactions, and more generally of Hamiltonians which do not belong to the available control setup. As a result, the implementation of $\hb(t)$ is often tricky if even possible \cite{bason1,rice1,jarz1}. For this reason, different STAs have been developed which connect adiabatic states at different desired times, but completely deviate from the adiabatic states at intermediate times \cite{rice2,ibanez2012,torrontegui2014,martinez2014,deffner2014,li2016,Baksic2016}, therefore giving up the benefits of true adiabaticity.

Here, we propose a method for achieving fast adiabatic driving remaining close to the adiabatic path, without needing new unrealizable terms in the Hamiltonian. This works by modulating the original Hamiltonian of the system in time such that it effectively replicates the dynamics induced by $\hb$ without needing any additional control Hamiltonian on the quantum system.

In order to do so, we first resort to control-theoretic techniques to study how the matrix structure of the correcting field $\hb$ is related to the initial set of control Hamiltonians which constitute $H$. This shows that $\hb$ can always be emulated, to arbitrary precision, by introducing a suitable (fast) time dependence in the control parameters. Second, we identify a class of Hamiltonians, generalizing the set of real ones, for which the matrix structure of $\hb$ can be discussed on general algebraic grounds and always involves components which are not directly controllable. They can, however, be simulated arbitrarily well using existing protocols.

Building on these results, we describe the construction of an effective counterdiabatic (E--CD) field $\ha(t)$, which is a time-dependent combination of the initially available control Hamiltonians. The control functions in $\ha$ will be chosen to be oscillating fast with respect to the natural time dependence of $H$, and $\ha$ will be enforced to simulate the dynamics $\ub$ induced by $\hb$.
As a result, the E--CD evolution tracks the adiabatic path, being arbitrarily close at a set of sampling time points while slightly deviating at intermediate ones. The general idea of producing the $\ub$ dynamics approximatively by working with a set of available control terms was also pursued, yet with a different strategy, in Ref. \cite{opatrny201}.

One must take into account that the adiabatic requirement of slow drivings can be recast as the need of strong fields. With this in mind, we characterize the efficiency of the E--CD scheme in terms of final fidelity and duration for different strengths of the correcting Hamiltonian $\ha$ as compared with the uncorrected $H$. 
We show that the E--CD field, when acting as an auxiliary weak term in $H$, realizes a consistent speed up of the adiabatic evolution. 
The E--CD Hamiltonian also works as a stand-alone STA, giving even better results but at the price of completely losing true adiabaticity. The net result is the ability to generate entanglement or perform state transfer significantly faster which, when taking into account noise from an external environment, inexorably leads to higher fidelities.

After recalling the theory of CD driving in Sec. \ref{sec:cddriving}, the control-theoretic setting in presented in Sec. \ref{sec:control}, together with the results on the general matrix structure of the correcting Hamiltonian $\hb$. The effective CD method is introduced in Sec. \ref{sec:approximate}, where the derivation of the E--CD field is fully worked through for a single spin system. The results are then exemplified via a number of applications in Sec. \ref{sec:examples}, the first being the Landau-Zener-Majorana model \cite{landau,zener,majorana,nori}. We afterwards show that a two-qubit entangled state can be prepared with $99.9\%$ fidelity ten times faster with respect to purely adiabatic evolution, using an E--CD field with strength comparable to that of the original Hamiltonian. We further discuss the case of a three-level system dynamically undergoing a sequence of avoided crossings in the energy spectrum \cite{Theisen2017}. In this section, the E--CD method is also benchmarked against traditional finite-time adiabatic driving, and the main advantages and limitations are discussed in detail, before proceeding to the conclusion in Sec. \ref{sec:conclusion}.

\section{Counterdiabatic driving} \label{sec:cddriving}
In this section, we recall the central elements of the theory of counterdiabatic fields \cite{rice1,berry1} which underpins our method.

A unitary evolution $U(t)$ is always a solution of the Schr\"odinger equation with Hamiltonian $H_U(t)=i \hbar \partial_t U U^{\dagger}$. Equivalently, the same $H_U(t)$ is a reverse-engineered Hamiltonian producing dynamics $U(t)$ exactly. Our aim is to find this $H_{U}(t)$ for the $U(t)$ which performs perfect adiabatic transfer. By construction, the new Hamiltonian will give the desired dynamics in an arbitrarily short period of time.

Let the initial Hamiltonian of the physical system be 
$$H(t)=\sum_{n=1}^N E_n(t) \ket{n(t)}\bra{n(t)},$$
having instantaneous eigenvalues $E_n(t)$ and instantaneous eigenvectors $\ket{n(t)}$. Let $\ucd(t) = \sum_{n} e^{-i \varphi_n(t)/\hbar} \ket{n(t)}\bra{n(t_0)}$, with $\varphi_n$ arbitrary phases, be a unitary matrix which diagonalizes $H(t)$ at all times. That is,
\begin{equation}\label{eq:diagonalize}
\ucd^{\dagger}(t) H(t) \ucd(t) = \text{diag}\{E_1(t),\dots,E_N(t)\}.
\end{equation}
The corresponding reverse-engineered Hamiltonian $\hcd = i \hbar \partial_t \ucd \ucd^{\dagger}$ reads
$$\hcd(t) = \sum_n \partial_t \varphi_n(t) \ket{n(t)}\bra{n(t)} + \hb(t), $$
where we have introduced the CD field \cite{berry1} 
\begin{equation}
\hb(t) = i \hbar \sum_n \ket{\partial_t n(t)} \bra{n(t)}.
\end{equation}
The Hamiltonian $\hcd$ drives the instantaneous eigenvectors of $H(t)$ exactly, with relative phase factors $\varphi_n(t)$, $e^{-i \varphi_n(t)/\hbar} \ket{n(t)}$. In particular, one can choose $\varphi_n(t)=0$ for all $n$, which gives $\hcd=\hb$. Therefore, implementing just $\hb$ is sufficient if one is only interested in preserving the populations of the instantaneous eigenvectors. Another interesting choice is $\varphi_n(t) = \int^t E_n(t') dt'$, in which case $\varphi_n$ are the adiabatic dynamic phase factors. This gives $\hcd = H(t) + \hb$, so that $\hb$ can be interpreted as a correcting field acting beside the initial Hamiltonian $H(t)$.

Alternative expressions of $\hb$ are useful in order to highlight specific properties. For example, from
\begin{equation}\label{eq:hb2}
\hb = i \hbar \sum_{m\ne n}\sum_{n=1}^N \frac{\ket{m}\bra{m} \partial_t H \ket{n} \bra{n}}{E_m-E_n},
\end{equation}
one can see that $\hb$, after having absorbed the geometric phases into the $\varphi_n$ \cite{berry1}, is purely off-diagonal in the basis of instantaneous eigenvectors (adiabatic basis). In terms of the Hilbert--Schmidt inner product $\tr(A^{\dagger} B)$, this implies that $\hb$ is orthogonal to $H$, and in general to all matrices commuting with $H$. It is also orthogonal to $\partial_t H$. Moreover, if we assume that the whole evolution takes place starting from an initial time $t_i$ for a total time $\tau$, one can rescale time according to $s=(t-t_i)/\tau$ and see that $\hb$ scales like $1/\tau$. In other words, the faster the process, the stronger the field $\hb$ should be. This property will be important in our discussion and it is particularly interesting in the context of quantum speed limits \cite{Campbell2017,DeffnerCampbell2017}.

By taking the derivative of Eq. \eqref{eq:diagonalize}, one can obtain the relation
\begin{equation}\label{eq:HBgenerator}
\partial_t H(t) = \frac{i}{\hbar} [H(t),\hb(t)] + \partial D(t),
\end{equation}
where $\partial D(t) = \sum_{n=1}^N \partial_t E_n(t) \ket{n(t)} \bra{n(t)}$.
Equation \eqref{eq:HBgenerator} highlights the fact that $\partial D(t)$ generates the variation of the instantaneous eigenvalues of $H(t)$, while $\hb$ is responsible for the variation of eigenvectors. This is so since $[\partial D(t),H(t)]=0$, and then $\partial D$ makes $H(t)$ ``move'' inside the set of Hamiltonians which commute with $H$ at time $t$, without changing the eigenvectors. On the other hand, the term $i[H,\hb]/\hbar$ determines the deviation from zero of the off-diagonal elements of $H(t+dt)$ with respect to the instantaneous adiabatic basis at time $t$.

\section{Control framework} \label{sec:control}

In order to realize $\hb$, it is of great interest to understand how the structure (matrix components) of the Hamiltonian $\hb$ is related to the structure of the initial Hamiltonian $H$ and to the control resources. 
In this section we exploit techniques from control theory to show that, assuming that $\hb$ can be computed, its action can be approximated at all times, arbitrarily well, by a time-dependent tuning of the initial parameters in the system Hamiltonian. This proves that the realization of our E--CD scheme is actually possible, and the specific construction will be discussed in the next section, Sec. \ref{sec:approximate}.
Second, this study permits us to identify a class of Hamiltonians, that generalizes the set of real Hamiltonians, for which all matrix elements of $\hb$ are not directly accessible without implementing new terms, which are not present in the original Hamiltonian.

Let $H(t)$, of finite dimension $N$, be realized by tuning some available time-independent control Hamiltonians $\hh = \{H_1,\dots, H_M \}$ via a set of continuous control functions $\bm{u}(t) = \{u_1(t), \dots, u_M(t) \}$. That is, $H(t)$ can be expressed in the form 
$$H\{\bm{u}(t) \} = \bm{u}(t)\cdot \hh =  \sum_{k=1}^M u_k(t) H_k.$$ 
Let $-i \hh$ be the vector of skew-Hermitian matrices $\{-i H_1,\dots, -i H_M \}$.
The matrices $-i \hh$ generate the so-called dynamical Lie algebra $\mathcal L$ of the system \cite{dalessandro2007}. This is the smallest algebra which contains $-i \hh$, all possible commutators $[-i H_j, -i H_k]$ of matrices from $-i \hh$, all possible commutators of commutators, and so on, considering all possibly nested commutators $[-i H_l, [\dots,[-i H_j, -i H_k]]\dots ]$.

A basis of $\mathcal L$ can be constructed by calculating, as a first step, the commutator of all possible pairs of matrices drawn from $-i\hh$. Among the new obtained matrices, one should select those which are linearly independent from themselves and the original set, and compute the commutator of such new matrices with this original set. The procedure is repeated iteratively until no new linearly independent elements are produced. An explicit algorithm can be found in \cite{Schirmer2001}. In our context, the linear span $\spn$ can be though of as the set of all possible matrices attainable by $H(t)$ at different times, for different values of the control functions.
 
Dynamical Lie algebras are of central importance in the study of the controllability of quantum systems \cite{dalessandro2007,Dirr,Glaser2015}. This has its origin in the fact \cite{Jurdjevic} that the set of reachable states, i.e., the set of unitary matrices that can be obtained as solution of the controlled Schr\"odinger equation for different choices of the control functions, coincides with the connected Lie group generated by $\mathcal L$.
An intuition for this can be given as follows. If an element $- i H$ belongs to $\mathcal L$, then the Lie group element $e^{-i H t}$ that it generates at time $t$ can be realized, to arbitrary precision, by suitably concatenating group elements generated by Hamiltonians in the initial set $-i \hh$, which belong to the one-parameter subgroups $\{e^{-i H_1 t},\dots,e^{-i H_M t}\}$. From a control-theoretic perspective, this means that the evolution produced by $H$ in a time $t$ can also be obtained by a sequence of evolutions governed by the Hamiltonians $\hh$, in general in a different total time.

Lie algebraic methods were also used in Ref.s \cite{torrontegui2014,martinez2014}, in the context of STA, for designing feasible shortcuts connecting the same initial and final adiabatic states, but following different paths in the Hilbert space.

Assuming that all Hamiltonians involved are made traceless, $\mathcal{L}$ can at most be $\mathfrak{su}(N)$, the algebra of skew-Hermitian traceless matrices generating the group of special unitary matrices $SU(N)$, and has dimension $N^2-1$. When $\mathcal L=\su(N)$, the system is said to be (operator) controllable \cite{dalessandro2007}. The Lie group generated by $\mathcal{L}$ through the exponential map will be denoted by $e^{\mathcal L}$. 

Since $-i \hb(t)$ is a skew-Hermitian matrix, it must belong to $\su(N)$ at all times. Therefore, if $\mathcal L = \su(N)$ then $-i\hb(t)$ is obviously in $\mathcal L$. The first result we show is that $\hb \in \mathcal L$ is a general property of $\hb$, even when $\mathcal L$ is not equal to $\su(N)$, but is rather a smaller subalgebra $\mathcal L \subset \su(N)$ -- that is, even when the system is not fully controllable. This is an interesting result for two reasons. First of all, it allows us to restrict the class of Hamiltonians needed to identify and realize $\hb$, with respect to the full set of traceless Hermitian matrices. Second, and more importantly in the present work, it means that the action of $\hb$ can be always approximated, in the sense of the action of the Lie group, by working only with the initially available Hamiltonians.
This result is stated in the following theorem, whose proof is given in Appendix \ref{appendix:proofs}. In all the following results, we will generally assume that $\hb$ can be computed. This excludes, in general, cases in which evolving eigenvectors can be degenerate \cite{rice2}.

\begin{thm} \label{theorem1} Let the Hamiltonian of the system be expressible, at all times and for all values of the control functions $\bm{u}(t)=\{ u_1(t),\dots,u_M(t)\}$, as a linear combination $H(t) = \sum_{k=0}^M u_k(t) H_k$ of time independent control Hamiltonians $\hh=\{H_1,\dots,H_M\}$. Let $\mathcal L$ be the Lie algebra generated by the matrices $-i \hh$. Assuming the $\hb$ exists, then $-i \hb(t)$ belongs to $\mathcal{L}$, for all times $t$.
\end{thm}

When the Lie algebra $\mathcal L$ has an additional structure, the general form of $\hb$ can be characterized more accurately.
In order to do so, let us introduce Cartan decompositions, which are important tools in the study of controllability of quantum systems \cite{dalessandro2007,Dirr,KHANEJA2001}. These are decompositions of the algebra $\mathcal L$ into the direct sum form $\mathcal L = \mathfrak{h} \oplus \mathfrak{p}$ which satisfy the commutation relations \begin{equation}\label{eq:Cartancomm}
[\mathfrak{h},\mathfrak{h}]\subseteq \h, \quad [\h,\p]\subseteq \p, \quad [\p,\p]\subseteq \h.
\end{equation}
We can then state the following result (some of the quantities were defined in Sec \ref{sec:cddriving}).

\begin{thm} \label{theorem2}
Given a Cartan decomposition $\mathcal L=\h\oplus \p$, if $\spnn \in \p$ and $-i \partial D\in \p$ for all $t$, then $-i \hb(t) \in \h$ for all $t$.
\end{thm}
\begin{proof}
If $-i\hh \in \p$, then also $-i \partial_t H \in \p$, since it can be written as a linear combination $-i \partial_t H = -i \partial_t \bm{u}(t) \cdot \hh$. If also $-i \partial D \in \p$ for all $t$, then from Eq. \eqref{eq:HBgenerator} it holds $ [H,\hb] \in \p$. From commutation relations defining the Cartan decomposition, one can then conclude that $-i \hb \in \h$ for all $t$.
\end{proof}
\begin{figure}
\includegraphics[scale=0.7]{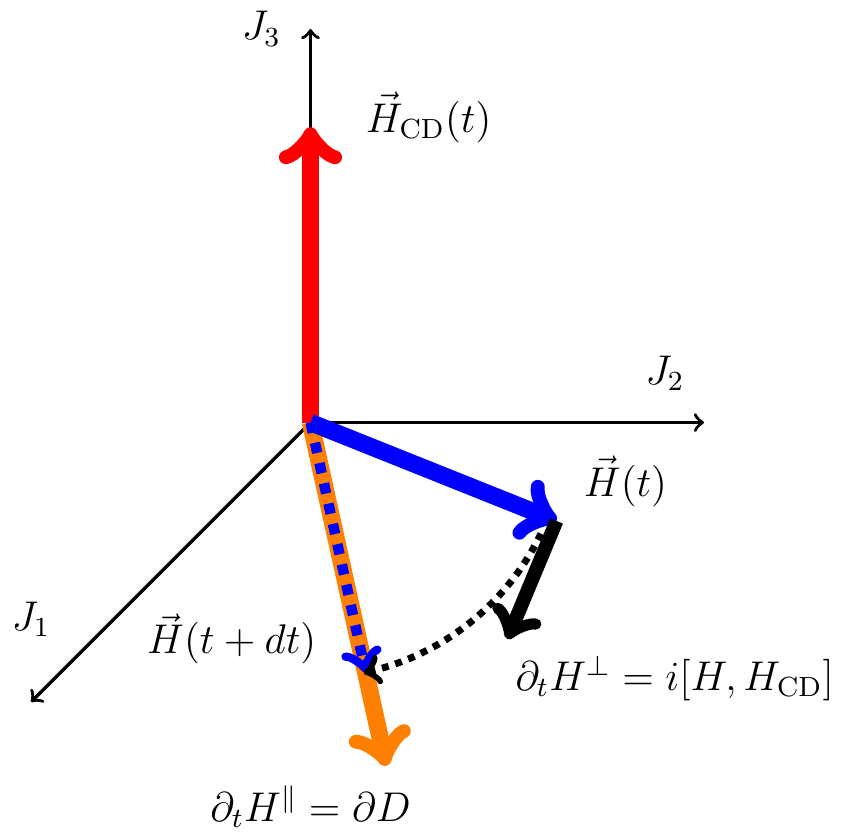}
\caption{Sketch of the geometric interpretation of relation \eqref{eq:HBgenerator} for the algebra $\mathfrak{su}(2)$. The Hamiltonian $H(t) = u_1(t) J_1 + u_2(t) J_2$ can be interpreted as a vector $\vec{H} = \{u_1,u_2,0\}$, whose magnitude $\sqrt{u_1^2+u_2^2}$ is the largest eigenvalue, in the three-dimensional space spanned by the $\su(2)$ generators $J_1,J_2,J_3$. The term $\partial D$ is responsible for variations of magnitude (eigenvalues), which we indicate with $\partial_t H^{\parallel}$. The term $i [H,\hb]$ generates rotations of the vector, i.e., a change of eigenvectors of $H$, and we denote this component of the variation $\partial_t H^{\perp}$.}
\label{fig:vecs}
\end{figure}
Let us clarify the result of Theorem \ref{theorem2} by means of two examples.
The first one is the situation in which $\hh$ includes exclusively real symmetric Hamiltonians. If this is the case, $H(t)$ can always be diagonalized by a (real) special orthogonal matrix $O(t)$. Thus $\hb(t) = i\partial_t O O^{\dagger}$ is necessarily a purely imaginary, skew-symmetric matrix, and $-i \hb(t)$ has thus no component on $-i \hh$ for all $t$.
When the set of control matrices generates the full algebra $\su(N)$, this corresponds to a Cartan decomposition of the form $\mathfrak{so}(N) \oplus \mathfrak{I}$, where $\mathfrak{so}(N)$ is the algebra spanned by the real matrices in $\su(N)$, while $\mathfrak{I}=\mathfrak{so}(N)^{\perp}$ is the set of purely imaginary matrices in $\su(N)$.

The second example is the angular momentum algebra $\mathfrak{su}(2)$ with $\hh$ containing two arbitrary generators $J_i,J_j$ among the three possible $\{J_1,J_2,J_3\}$. Explicit diagonalization and computation of $\hb$ \cite{rice1,berry1} shows that then $\hb$ is proportional to the third generator.
This can be also seen from Eq. \eqref{eq:HBgenerator}: interpreting $H$ geometrically as a vector $\vec{H}(t) = \{H_1,H_2,H_3\}$ in the three-dimensional algebra, the matrix $\partial D$ generates the stretching of the Hamiltonian vector along in its own direction, while $\hb$ generates rotations on the plane spanned by $J_i$ and $J_j$; see Fig. \ref{fig:vecs}. Therefore, recalling the properties of the rotations in 3D, $\hb$ must lie in the direction of $J_k$. In this case the Cartan decomposition is unitarily conjugate to $\mathfrak{so}(2)\oplus \mathfrak{I}$.
\section{Effective CD field} \label{sec:approximate}
From Theorem \ref{theorem1} one knows that (i) $\hb$ can be approximated by a suitable choice of the control functions $\bm{u}(t)$ for all $t$ and that (ii) $\hb$ can be expressed as a linear combination of (nested) commutators of elements of $\hh$. Observation (i) suggests that the problem of finding an effective CD field can be formulated as that of finding a correcting Hamiltonian $\ha$, of fully controllable form 
\begin{equation}\label{eq:ha}
\ha(t) = \bm{c}(t)\cdot \hh = \sum_{k=1}^M c_k(t) H_k,
\end{equation}
which emulates $\hb$. More precisely, the Hamiltonian $\ha$ should produce a dynamics $\ua$ emulating the one induced by $\hb$, i.e., $\ub$.
While the algebraic structure suggests this can be done, it does not constructively specify what the $c_k(t)$ need to be. To find them, from observation (ii), we choose to adopt a representation of the propagators $\ua$ and $\ub$ based on a Magnus expansion, see Appendix \ref{appendix:Magnus}, where terms involving commutators of the Hamiltonians appear naturally. The first two terms are given for completeness,
\begin{align*}
M^{(1)}(t) &= -\frac{i}{\hbar} \int_0^t H(t_1) dt_1, \\
M^{(2)}(t) & = \left(\frac{-i}{\hbar}\right)^2 \int_0^{t} dt_1 \int_0^{t_1} dt_2 [H(t_1),H(t_2)].
\end{align*}
The first step for designing the effective field $\ha$ is to choose an ansatz for the control functions $\bm{c}(t)$ involving a certain number of free parameters. The second, considering a small evolution time $t$, is to ask the first terms of the Magnus expansion of $\ub$ and $\ua$ to coincide up to a desired order in $t$. This requirement will produce constraint equations for the free parameters.  
Since $\ub$ and $\ua$ cannot coincide at all times, one needs to choose a discrete temporal grid $\{t_k\}$ for decomposing the whole evolution, and enforce the approximate equality at such time points. As a result, the approximating dynamics will match the true one at times $\{t_k\}$, while deviating at intermediate times.  

The matrix components of $\hb$ which do not belong to $\spn$ will not appear in the first Magnus term of $H+\ha$, which is essentially the time average of $H+\ha$. They will appear from the commutators in the following terms. Therefore, if one wants $\ha$ to reproduce $\hb$ effectively, it must hold that $\ha$ gives no contribution to the first term, i.e., has vanishing time average on each interval $(t_k,t_{k+1})$. Besides, since the mediated effect will be of some order $m$ greater than one in $t$, the magnitude of the control functions in $\ha$ must be proportional to $t^{-X}$, for some $X>0$ depending on $m$, in order to amplify the mediated effect so that it acts at first order. 

Making these general ideas concrete, we choose the control functions in Eq. \eqref{eq:ha} to have the form of a truncated Fourier series,
\begin{equation}\label{eq:contrfunc}
c_k(t) = \omega^X \sum_{j=1}^L \big[ A_{k,j} \sin(j \omega t) + B_{k,j} \cos(j \omega t)],
\end{equation}
involving at most $L$ harmonics of the fundamental frequency $\omega$.
This way, the needed time grid is naturally defined by the set of stroboscopic times $t_k = t_i+k T$, $k\in \mathbb{N}$, where $T = 2 \pi/\omega$ is the period of the control functions and $t_i$ the initial time. The constant amplitudes $\{A_{k,j},B_{k,j}\}$ are the free parameters used to enforce the constraints on the Magnus expansion at the end of each period. As a result, the control functions so determined will be discontinuous between different periods, due to the jumps in the values of the amplitudes. From a practical perspective, this is of course physically inconvenient, and so one would like to come up with an interpolation providing continuous and possibly smooth functions $A_{k,j}(t), B_{k,j}(t)$. We will see that this can indeed be done, and often in a natural way.

The smaller the period $T$, the better the target dynamics $\ub$ will be sampled, so in general $\ha$ will need to oscillate fast with respect to the time-dependence of $H$. Raising the number of harmonics $L$ permits us to introduce more parameters, and thus to obtain and solve constraint equations produced by higher Magnus terms. Ideally, one would like to find a good compromise between high sampling rates of $\ub$, good approximation at stroboscopic times, while keeping the fundamental frequency and the number of harmonics as small as possible. This will be discussed in detail in Sec. \ref{sec:examples}. 
A strategy based on similar ideas was introduced in Ref. \cite{Verdeny2014} for the purpose of producing effective time-independent Hamiltonians via optimal control.

Let us now focus on the class of Hamiltonians identified by Theorem \ref{theorem2}. Due to the commutation relations from Eq. \eqref{eq:Cartancomm}, the Magnus terms of $\ua$ involving an odd number of commutators, $\ma^{(2k)}$, belong to $\h$. Those involving an even number of them, $\ma^{(2k+1)}$, belong to $\p$ instead. Therefore, the general form of the constraint equations is
\begin{align*}
& \ma^{(2k)} =\mb^{(k)}, \\
& \ma^{(2k+1)} = 0.
\end{align*}

A discussion of the accuracy of the method is presented in Appendix \ref{sec:infiderror}, where an estimate of the expected error in the probability of not being in the target state, at the end of one period, is derived as a function of the number of solved constraint equations.
Let us show the overall procedure by working through a specific case.

\subsection{Worked through case: single spin} \label{sec:su2}
We concentrate here on the case of the Lie algebra $\mathfrak{su}(2)$, which can physically describe, for example, a single spin driven by magnetic fields. A first application of STA methods via CD driving to this kind of system was studied in Ref. \cite{Chen2010}. Since for tracking the instantaneous ground state it is sufficient to implement $\hb$ without $H$, as discussed in Sec. \ref{sec:cddriving}, and for simplicity of presentation, we treat the case in which one wants to approximate $\hb$ alone by means of $\ha$. The general case, $H+\hb$ approximated by $H+\ha$, will be discussed at the end of this section. For clarity, let us work with the Pauli matrices $\{\sx,\sy,\sz\}$, having commutation relations $[\sigma_i,\sigma_j] = 2 i \epsilon_{ijk} \sigma_k$.
We assume that the Hamiltonian of the system is controllable only along two directions, say $x$ and $z$, and can thus be written as:
\begin{equation}
H(t) = u_x(t) \sx + u_z(t) \sz.
\end{equation}
The field $\hb$, as discussed in Sec. \ref{sec:control}, will be directed along the unavailable direction. It can be calculated explicitly \cite{rice1,berry1} and it has the form $\hb(t) = \fb(t) \sy$, with:
\begin{equation}
\fb(t) = -\frac{1}{2} \frac{u_x(t) \partial_t u_z(t) - \partial_t  u_x(t) u_z(t)}{u_x(t)^2+u_z(t)^2}.
\end{equation}
To achieve compensation to first order in $T$ at stroboscopic times, we apply the E--CD Hamiltonian $\ha = c_x(t) \sx+c_z(t) \sz$, with control functions of the form \eqref{eq:contrfunc}. The first non-zero term of the Magnus exponent produced by $\ha$ is the second one, and has the desired matrix structure.   Generalizing to the choice of two arbitrary generators $\{\sigma_i,\sigma_j\}$, one has
\begin{equation} \label{eq:Ma2}
\ma^{(2)} (T) = - i\epsilon_{ijk} \frac{2\pi}{\omega^{2-2X}} \sum_{n=1}^L \frac{1}{n} \left[A_{i,n} B_{j,n}-B_{i,n} A_{j,n}  \right] \sigma_k .
\end{equation}
From Eq. \eqref{eq:Ma2} one can see that the highest order in $T$ is produced by the interplay of $\sin$ and $\cos$ components related to the same harmonic $n$, belonging to different control functions $c_i$ and $c_j$. There is no mixing between different harmonics, and no mixing between fields with the same phase. This feature is due to the properties of the integrals of the form $\int_0^T dt_1 \sin(n \omega t_1) \int_0^{t_1} dt_2 \cos(k \omega t_2)$.
Higher-order Magnus terms can be computed analytically, but the expressions get longer and longer with respect to Eq. \eqref{eq:Ma2}. 

Now, let us find the constraint equation to first order in $T$. For this purpose, two amplitudes are sufficient and, since we want $\ma^{(2)}$ to be of order $T$, we choose $X=1/2$ in Eq. \eqref{eq:contrfunc}.
A possible choice of control functions is then
\begin{equation}
c_z(t) = A \sqrt{\omega}\sin(\omega t); \quad c_x(t) = B \sqrt{\omega} \cos(\omega t).
\end{equation}
Using this functions, we can compute the first Magnus terms. Let $t_n=t_i+nT$, $n=1,2,\dots$, and let us indicate with $\int^T$ the integral over a full period $\int_{t_n }^{t_n+T}$. The exponents of the Magnus expansion for the CD and the effective dynamics at the end of each period $n T$ of the control functions are, respectively,
\begin{align}
\mb(T) & = -i \left( \int^T \fb(t) dt \right) \sy \nonumber, \\
& = -i [\fb(t_n+T/2) T + o(T^3)]\sy\label{eq:mb1},\\
\ma(T) & = -i AB T \sy + o(T^{3/2}) \label{eq:1ordsu2}.
\end{align}
The last equality in Eq. \eqref{eq:mb1} is obtained by formally Taylor-expanding $\fb(t)$ around the midpoint of the integration interval, $t_n+T/2$, and then integrating.
Equating order-$T$ terms in Eqs. \eqref{eq:mb1} and \eqref{eq:1ordsu2} one obtains the first constraint equation: $AB = \fb(t_n+T/2)$. This can be straightforwardly solved, and a possible solution for the $n$-th period is
\begin{equation}\label{eq:amplitudes1}
A = \sqrt{|\fb(t_n+T/2)|}, \quad B=\sign\big[\fb(t_n+T/2)\big] A,
\end{equation}
where $\sign(x)$ indicates the sign of $x$ and takes care of the case in which $\fb(t)$, and thus the product $AB$, is negative. Different solutions may be more suitable depending on the structure of $f_\text{CD}$.

The fact that the above solution for $A$ and $B$, Eq. \eqref{eq:amplitudes1}, is a simple function of $\fb$ evaluated at the midpoint of the period suggests a straightforward manner to interpolate the solutions in all intervals, obtaining continuous and smooth amplitude functions $A(t)$ and $B(t)$. It is indeed sufficient to replace $\fb(t_n+T/2) \to \fb(t)$ in Eq. \eqref{eq:amplitudes1} and the same accuracy is maintained. As a conclusion, the effective field can be written as
\begin{equation}\label{eq:1stACD}
\ha(t) = \sqrt{|\fb(t)| \omega} \Big[ \sign\{\fb(t)\}\cos(\omega t) \sx + \sin(\omega t) \sz \Big].
\end{equation} 
Two immediate observations can be made from Eq. \eqref{eq:1stACD}: First, $\ha$ is proportional to $\sqrt{\omega}$, so, as expected, a better sampling can be obtained at the cost of having a stronger field. Second, due to the proportionality with $\sqrt{|\fb(t)|}$, the effective CD field inherits to some extent the behavior of the exact one. In particular, it vanishes when $\hb$ does, that is, when the probability of nonadiabatic transitions is zero.

The third constraint equation, $\ma^{(3)}=0$, can be solved by adding a further term $-4 A \sin(2 \omega t) \sz$ to Eq. \eqref{eq:1stACD}. 

We conclude this section by discussing the more interesting case in which $H+\ha$ approximates $H+\hb$. This scenario is particularly important in a quasiadiabatic regime. When nonadiabatic effects are strong, the correcting fields become dominant with respect to the initial Hamiltonian $H$, and the method is highly nonadiabatic. In fact, even if the instantaneous eigenvectors of $H$ are tracked during the evolution, the system is not in an instantaneous eigenstate of the total (corrected) Hamiltonian which is actually applied, $H+\hb$ (or $H + \ha$). On the other hand, if the nonadiabatic effects are weak, then one can think of being close to true adiabaticity.

Depending on the problem, it might be convenient to move to the time-dependent interaction picture for computing the effective field. In general, though, it becomes more complicated to compute the terms in the Magnus expansion. In any case, the treatment above can be repeated, and to order $T^{3}$ in infidelity, one still obtains the same effective field as in Eq. \eqref{eq:1stACD}.

\section{Applications} \label{sec:examples}

In this section, we discuss three applications of the method introduced in Sec. \ref{sec:approximate}. The first one is the Landau-Zener-Majorana (LZM) model  \cite{landau,majorana,zener,nori}. This is the standard milestone in the study of nonadiabatic effects and more generally of the physics near avoided crossings in the energy spectrum.
The second one regards the adiabatic preparation of a two-qubit entangled state, a task of central importance for quantum information processing \cite{NielsenChuang}.
The last application generalizes the first one and deals with a three-level system whose natural evolution induces the formation of a sequence of avoided crossings \cite{Theisen2017}. Such a model can describe local parts of many-body nontrivial energy spectra.

We also benchmark the efficiency of the E--CD method against standard (finite time) adiabatic driving. This is done with a focus on the LZM case, it being the simplest nonadiabatic scenario. This analysis underpins the usefulness of the method not only for theoretical purposes but, more importantly, for practical applications in the control of quantum systems.

For setting up the comparison between the E--CD and the adiabatic paradigms, a quantification of their respective performance is first needed. Since this is not an obvious task, let us here introduce the criteria which are adopted for this purpose.
The principle figure of merit we are interested in is the infidelity $\If = 1-|\braket{\psi(t_f)\lvert gs(t_f)} \vert^2$, i.e., the probability that the final state $\ket{\psi(t_f)}$ of the system is not in the instantaneous ground state $\ket{gs(t_f)}$ at the end of the protocol. We are also interested in making the evolution as fast as possible with as few resources as possible, but there is not a unique way to quantify this. The general idea is that one should not only take into account the full duration of the protocol, but also the intensity of the Hamiltonian which is applied to the system. Indeed, the adiabatic theorem \cite{messiah1961qm,Jansen2007} (see Appendix \ref{appendix:proofs}) formally states a property of the solutions of the family of Schr\"odinger equations 
\begin{equation}\label{eq:rescaledSchr}
i\hbar \partial_s U_\tau(s) = \tau H(s) U_\tau (s),
\end{equation}
for varying $\tau$. Therefore, the theorem can be interpreted, on one side, as describing a varying-duration behavior, if Eq. \eqref{eq:rescaledSchr} is obtained after a rescaling $s=t/\tau$ of the physical time for a fixed Hamiltonian. Alternatively, it can be interpreted as describing an intensity-varying behavior, if Eq. \eqref{eq:rescaledSchr} is obtained from considering the duration fixed and amplifying the Hamiltonian like $H(t)\to \tau H(t)$.

These considerations lead us to study the infidelity $\If$ both as a function of the protocol duration and of the ratio between the strength of $\ha$ and that of the initial Hamiltonian $H(t)$. We formalize this by choosing, as a measure of strength $\mathcal S(\cdot)$, the maximal Frobenius norm over the whole evolution,
\begin{equation}\label{eq:strength}
\mathcal S(H) \equiv  \max_t \norm{H(t)}_F .
\end{equation}
The Frobenius norm is defined as $\norm{A}_F=\tr(A A^{\dagger})$  and provides a quantification of the magnitude of the whole Hamiltonian matrix. Different choices are possible as cost functions depending on the resource one is interested in. For instance, one might be practically interested in the maximal matrix element, or the maximal or average amplitude reached by the control functions.

With these tools at hand, let us now proceed to the discussion of the above mentioned applications.

\subsection{Landau-Zener-Majorana model} \label{sec:LZM}
\begin{figure}[t!]
\includegraphics[width=\linewidth]{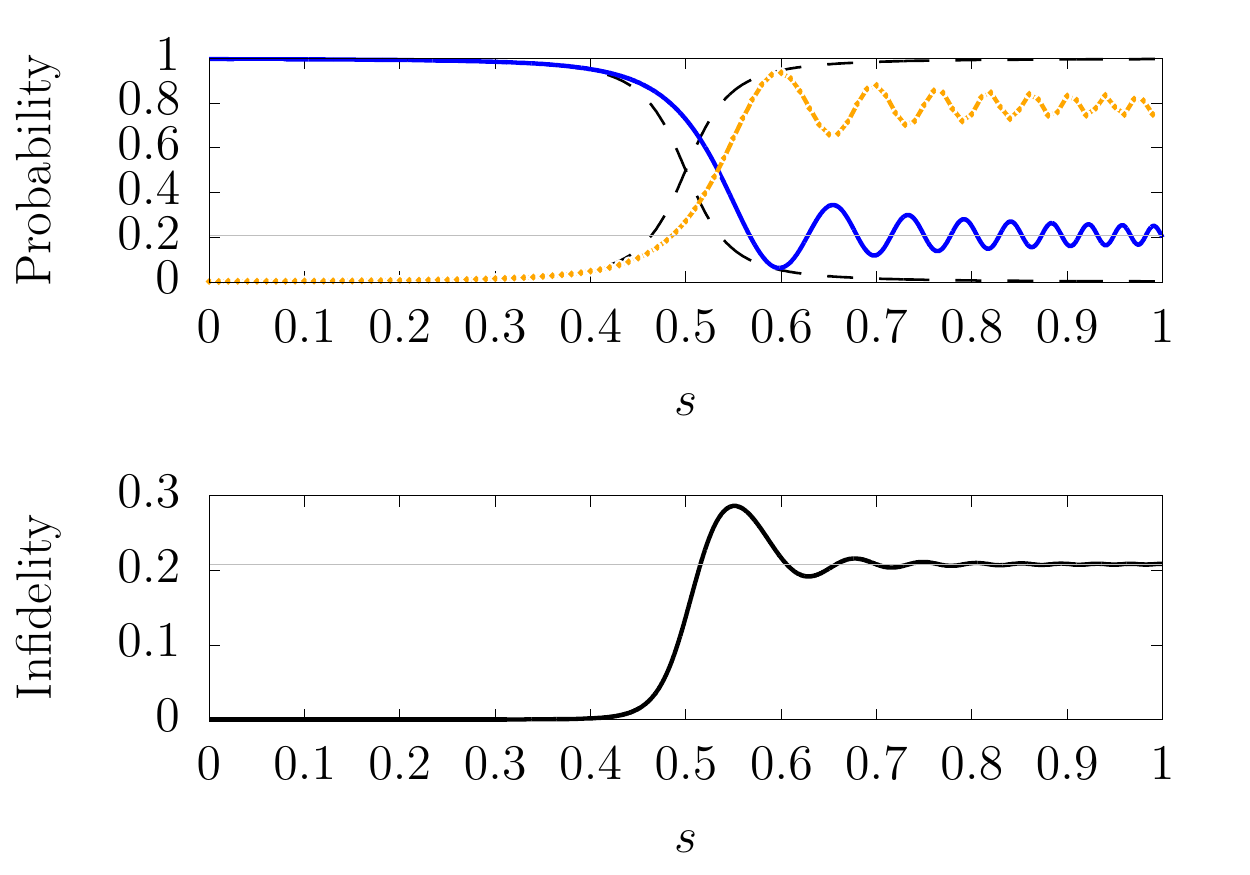}
\caption{Evolution given by the LZM Schr\"odinger equation \eqref{eq:lzham}, with parameters $\varepsilon = 20$, $\tau= \varepsilon$. Above, the populations of the bare states are depicted. The dashed lines represent the true instantaneous eigenstates, while the solid blue and dot-dashed orange lines represent the LZM evolution. Below, the infidelity, i.e., the probability of nonadiabatic transition, is shown. In both plots, the thin gray line represents the prediction given by the LZM formula, Eq. \eqref{eq:lzformula}.}
\label{fig:lzdynamics}
\end{figure}
We demonstrate the general approach detailed in Sec. \ref{sec:su2} for the specific case of the Landau-Zener-Majorana model \cite{landau,majorana,zener,nori}. Such a model describes a linear sweep of the gap between the energy levels of a two-level system. As they approach each other,
the presence of a coupling $\hbar \beta/2$ prevents a net crossing, and an avoided crossing is produced instead with minimal gap $\hbar \beta$. Assuming that the sweep spans an energy difference $\hbar \alpha$ in a time interval $t_i \le t \le t_f$, and that the anticrossing takes place at the intermediate time $t_c=(t_f -t_i)/2$, the Hamiltonian can be written in the form
\begin{equation*}
\hlz(t)/\hbar =  \frac{\alpha}{2} \frac{t-t_c}{t_f-t_i} \sz + \frac{\beta}{2} \sx. 
\end{equation*}
\begin{figure}[t]
\includegraphics[width=\linewidth]{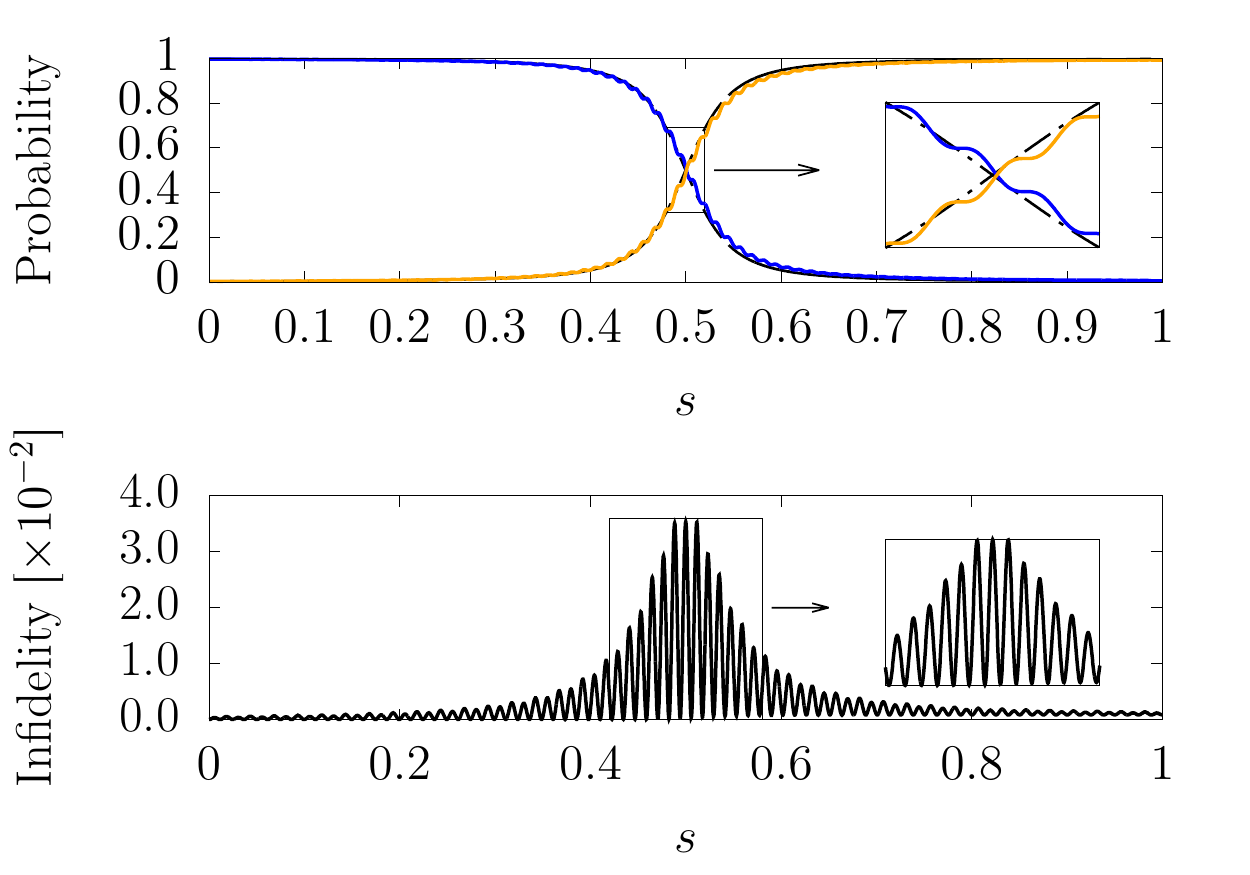}
\caption{Evolution given by the E--CD field, Eq. \eqref{eq:contham}, with parameters $\varepsilon=20$, $\tau=20$, $N_T= 2 \tau$. Above, the populations of the bare states are shown, with an inset zooming around the avoided crossing. The solid (orange/blue) lines represent the E--CD dynamics, which oscillates around the true adiabatic one, represented by dashed lines. Below, the infidelity, i.e., the probability of nonadiabatic transition, is shown.}
\label{fig:contdynamics}
\end{figure}
In order to study the adiabatic properties of the system, it is convenient to rescale the quantities in terms of $\beta$, which defines a fundamental frequency scale of the problem. Furthermore, let us parametrize the time according to $ s = \frac{t-t_i}{t_f-t_i}$, with $0 \le s \le 1$. In terms of the dimensionless quantities $s$, $\varepsilon = \alpha/\beta$ and $\tau=\beta (t_f-t_i)$, the Schr\"odinger equation becomes 
\begin{equation}\label{eq:lzham}
i \frac{\partial U(s)}{\partial s} = \frac{\tau}{2} \left[\varepsilon \left(s-\frac{1}{2}\right) \sz + \sx \right] U(s).\end{equation}
The parameter $\varepsilon$ determines the distance, in units of $\beta$, from the avoided crossing at the beginning and at the end of the protocol. It also rules the decay rate of the LZM oscillations in the transition probability, i.e., of nonadiabatic transition, after the anticrossing \cite{vitanov1,vitanovgarraway} (see also Appendix A in \cite{limberry}). Assuming that the system starts in the ground state, the transition probability to the excited state (i.e., of a nonadiabatic transition) is shown in Fig. \ref{fig:lzdynamics}. The asymptotic value, for large times, is given by the Landau-Zener formula \cite{zener,majorana,nori},
\begin{equation}\label{eq:lzformula}
P_\text{LZ} = \exp\left(- \frac{\pi \tau}{2 \varepsilon} \right).
\end{equation}
The CD field for this problem, satisfying $i \partial_s U = \tau H_\text{CD}(s) U$, is $H_\text{CD}(s) =  f_\text{CD}(s) \sigma_y$, with
\begin{equation*}
f_\text{CD}(s) = -\frac{1}{2 \tau} \frac{\varepsilon}{\varepsilon^2\left(s-\frac{1}{2} \right)^2+1}.
\end{equation*}
This CD protocol was experimentally studied in Refs. \cite{bason1,malossi1}.
From the results of the previous Sec. \ref{sec:su2}, Eq. \eqref{eq:1stACD}, the dimensionless effective CD field, satisfying $i\partial_s U = \tau H_\text{E}(s) U$,  is 
\begin{equation} \label{eq:contham}
H_\text{E}(s) = \sqrt{ |f_\text{CD}(s)| \omega}[-\cos(\omega s \tau) \sigma_x + \sin(\omega s \tau)\sigma_z ],
\end{equation}
with $\omega = 2 \pi/(\beta T)$.
The negative sign in front of the cosine is due to the fact that $\fb$ is always negative.

The dynamics given by $H$ alone can be compared with the one produced by $\ha$ from Figs. \ref{fig:lzdynamics} and \ref{fig:contdynamics}, for $\beta = 1$. As expected, the E--CD dynamics oscillates fast, remaining close to the target one. In the vicinity of the anticrossing, where nonadiabatic effects are strong, the deviation between consecutive periods is larger.

Let us now study the comparison between the E--CD and finite-time adiabatic methods more thoroughly. In particular, we show that the E--CD strategy can provide a consistent speed up of adiabatic protocols, characterizing the important regimes of parameters in terms of duration of the protocol and strength of the involved Hamiltonians. 
\begin{figure}[t]
\includegraphics[width=\linewidth]{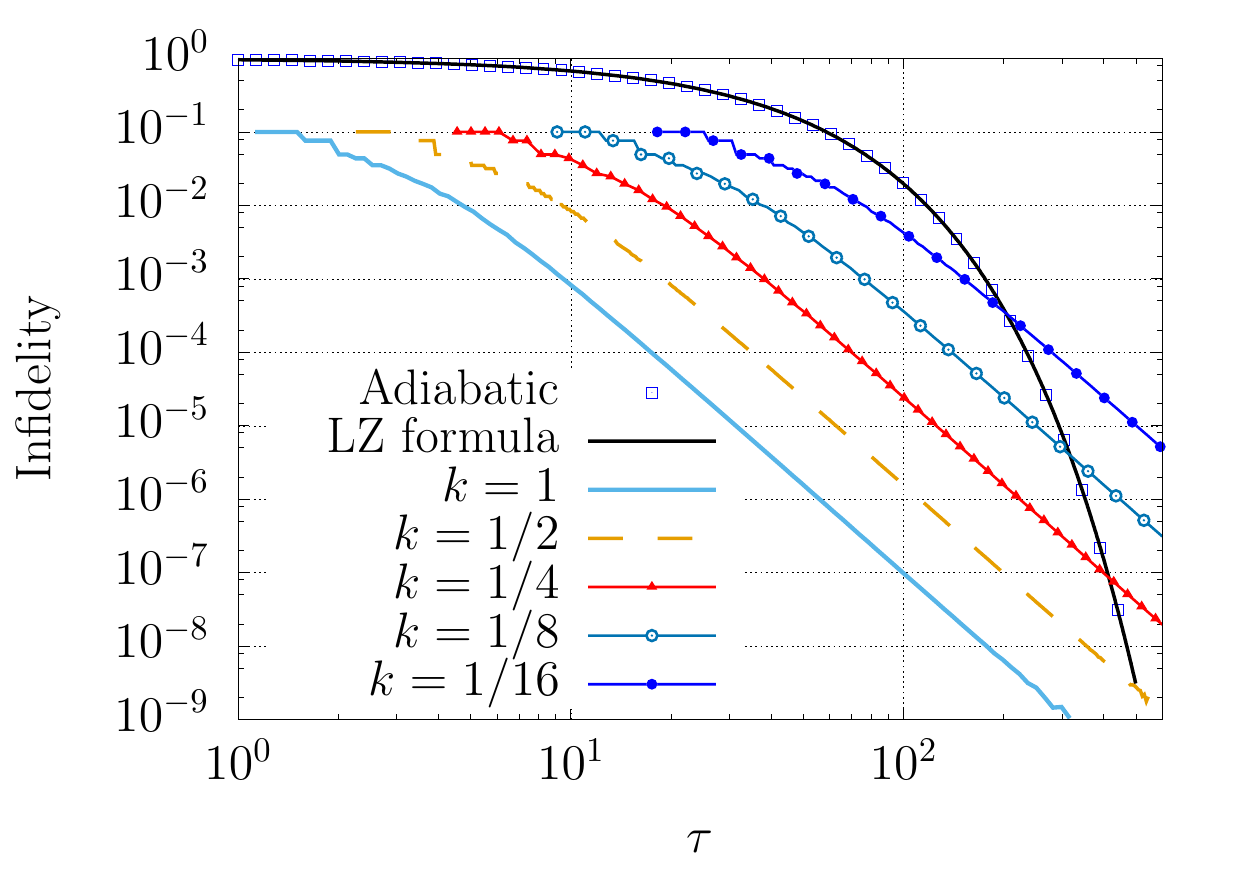}
\caption{Comparison between finite-time adiabatic dynamics and effective CD driving, using only $H_\text{E}$. Different (colored) line styles represent the results for different ratios $k$ in relation \eqref{eq:rationorms}. The empty squared points represent the adiabatic dynamics, whose behavior is well predicted by the LZ formula of Eq. \eqref{eq:lzformula} (black solid line). For each $\tau$, the number of periods $N_T$ is determined as the largest integer smaller than $\omega \tau/2 \pi$, once $\omega$ is chosen to be the largest such that Eq. \eqref{eq:rationorms} is satisfied. The conversion to integers explains the steps in the colored curves. The beginning of the latter is determined by the condition $N_T>1$.}
\label{fig:vstau}
\end{figure} 
 
Specifically, we compute the infidelity obtained with the E--CD Hamiltonian for different total durations $\tau$, after having selected the maximal frequency $\omega$ such that, for a certain factor $k$, the inequality
\begin{equation} \label{eq:rationorms}\mathcal S(\ha) \le k \mathcal S(H),
\end{equation}
with $\mathcal S$ defined in Eq. \eqref{eq:strength}, holds.

Let us remark that nonadiabatic transitions are exponentially weak in the adiabatic parameter \cite{dykhne1962,davis1976}. On the other hand, one expects that the efficiency of the E--CD method increases polynomially when raising the frequency $\omega$, since it is based on a perturbative argument. Therefore, one can safely predict that the adiabatic method, for very large $\tau$ and for a fixed maximal strength of the control fields, performs better. Nonetheless, what we are interested in is the intermediate range of durations, very far from the asymptotic limit in $\tau$, which is the regime where practical protocols need to work. 
\begin{figure}[t!]
\includegraphics[width=\linewidth]{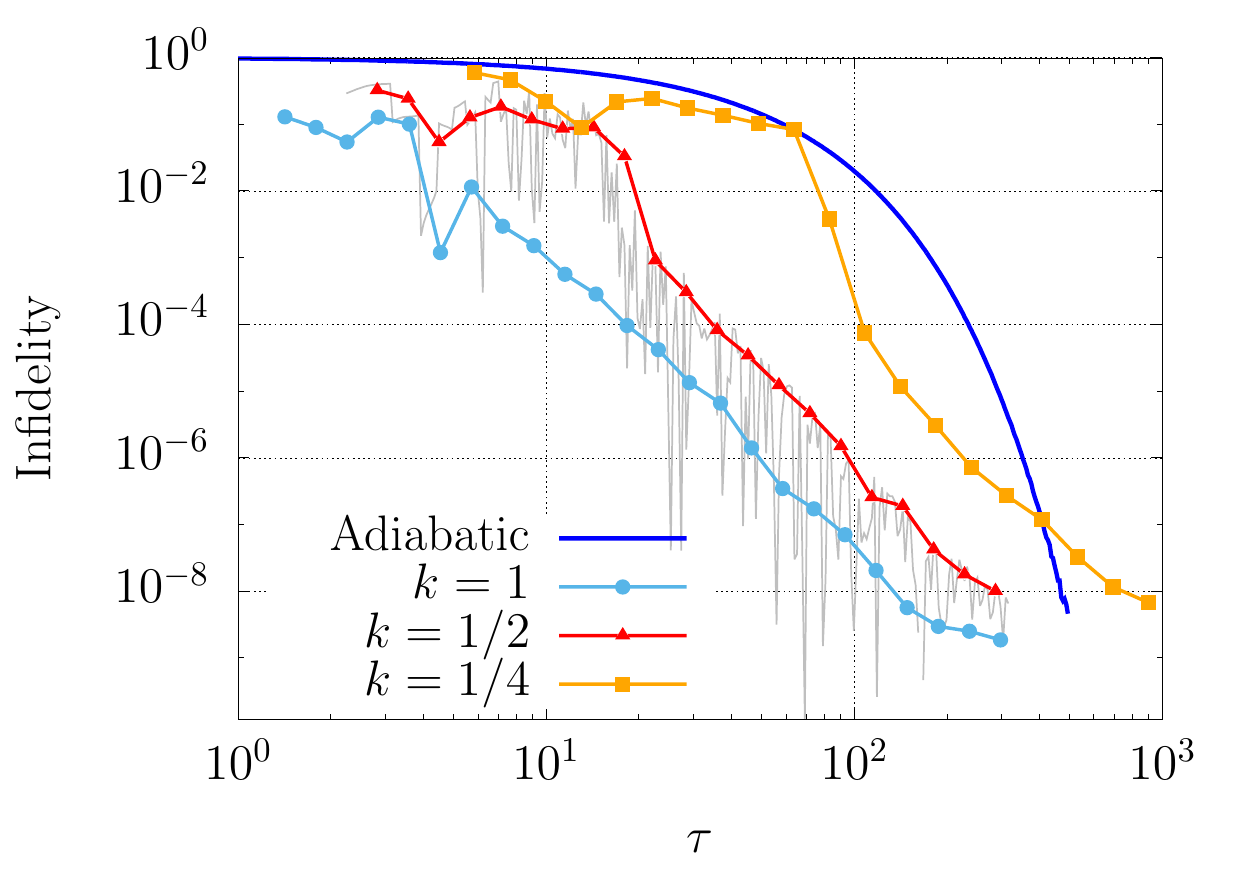}
\caption{The results obtained with the procedure described in Fig. \ref{fig:vstau}, shown for the application of the total Hamiltonian $H+\ha$. The solid line represents the adiabatic dynamics, while different line styles represent the E--CD method for different values of the factor $k$ in Eq. \eqref{eq:rationorms}, i.e., for different strength of $\ha$ as compared to $H$. The infidelity oscillates fast for different $\tau$, as shown by the thin gray line for the case $k=1/2$, so the depicted points are representing an average of 20 surrounding points.}
\label{fig:vstauontop}
\end{figure}

The results for the LMZ model are shown in Figs. \ref{fig:vstau} and \ref{fig:vstauontop}. In Fig. \ref{fig:vstau}, the case in which only $\ha$, as a stand-alone shortcut to adiabaticity, is applied is reported. As expected from the above discussion, the infidelity under effective CD driving scales like a power-law for large durations $\tau$, and eventually the adiabatic exponential decay takes over. However, in the intermediate range of $\tau$, the effective method does permit us to achieve an important speed-up of the adiabatic process, for fixed infidelity. The range of improvement gets wider and wider as one allows a stronger correcting field, while still remaining below the maximal norm of $H$. For example, one can reach a fidelity of $\sim 99.9\%$ twenty times faster using an E--CD field of the \emph{same} strength of $H(t)$ [$k=1$ in Eq. \eqref{eq:rationorms}].
Figure \ref{fig:vstauontop} shows the same results for the case in which the E--CD Hamiltonian is implemented beside $H(t)$, and thus acts as a weak correcting field to the major dynamics induced by $H(t)$. The final infidelity oscillates strongly for different values of $\tau$ (thin gray line), so the points indicate an average of 20 surrounding points. The oscillations can bring the infidelity to much lower values, but never much higher than those represented by the points. This method performs worse in numbers than the previous case (Fig. \ref{fig:vstau}) for small $\tau$, but its main advantage is that it allows one to preserve true adiabaticity to a certain extent. That is, the driven states are close to being instantaneous eigenstates of the full Hamiltonian $H+\ha$, since $\ha$ is weaker than $H$ and averages to zero in time. One can see that there is an initial regime in which the protocol is too fast for the correction to give a significant improvement, and a second regime in which the infidelity follows essentially the same behavior as in Fig. \ref{fig:vstau}.
As an example, an infidelity of at least $\sim 10^{-4}$ can be achieved with a speed up of around $\sim 6.5$ times, using a maximal strength of the correcting field satisfying Eq. \eqref{eq:rationorms} with $k=1/2$. A speed up of $2.2$ times is obtained in the case $k=1/4$ for the same infidelity.
All these results confirm that the E--CD method represents an efficient strategy for obtaining fast adiabatic evolution, without requiring strong drivings.

The second criterium for comparing adiabatic and effective CD driving is the following: we consider the total duration to be fixed, $0 \le s \le 1$, and we compare the dynamics produced by the amplified Hamiltonians $H_\tau(s)$ with the one produced by $\ha$ for increasing values of the frequency $\omega$. Due to the intertwined relation between duration and amplitude, we choose as a basis for comparison the integral of the norm over the whole evolution
\begin{equation}\label{eq:intfrob}
\int_{t_i}^{t_f} dt \norm{H_\text{X}(t)}_F .
\end{equation}
The results are shown in Fig. \ref{fig:intnorm}. As for the previous criterium adopted, they are rather promising, confirming the validity of the E--CD method. Indeed, the E--CD Hamiltonian turns out to always produce a better infidelity, for a given total integral norm of the Hamiltonian, in the range of values studied.

\begin{figure}
\includegraphics[width=\linewidth]{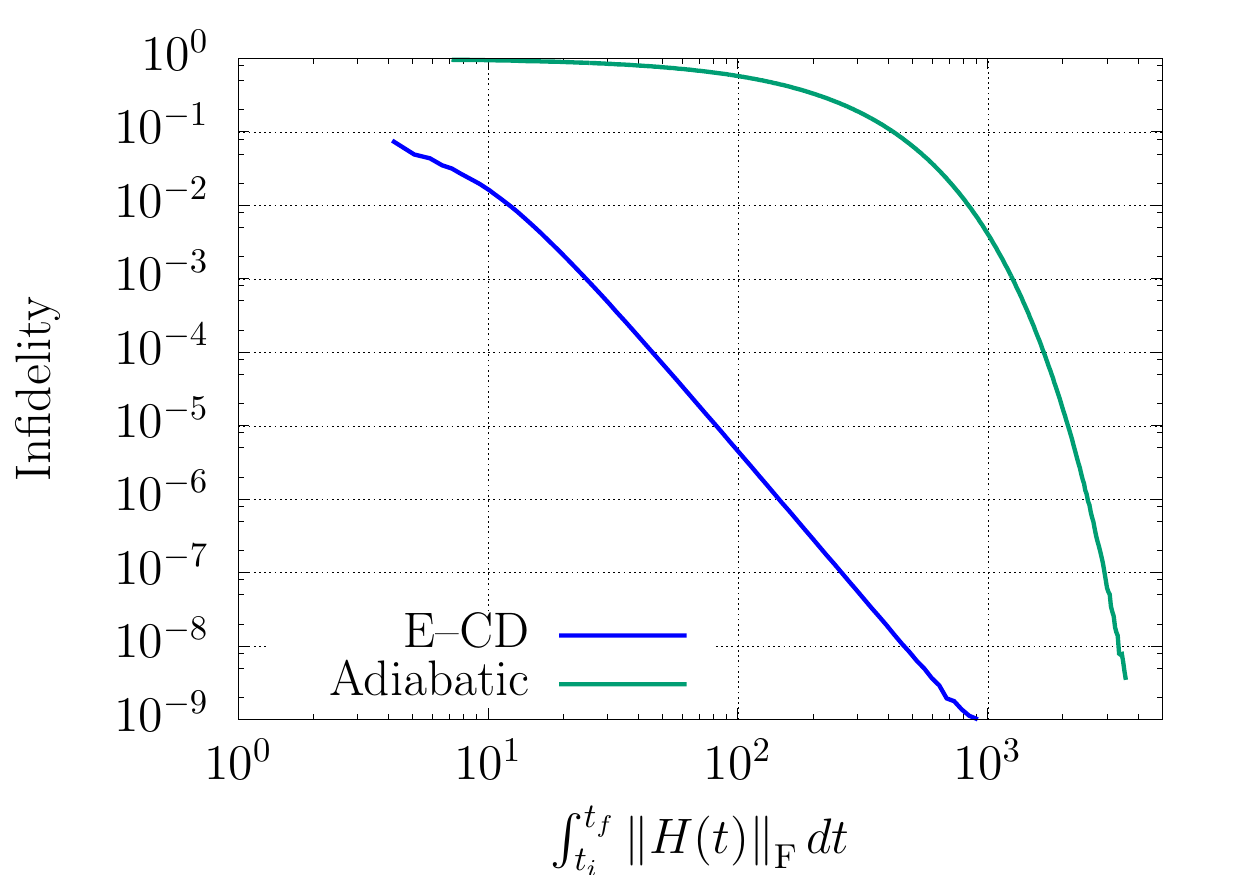}
\caption{Comparison between finite-time adiabatic and E--CD driving in terms of the infidelity as a function of the time integral of the Frobenius norm of the respective Hamiltonians [Eq. \eqref{eq:intfrob}], for parameters $\varepsilon=40, \beta=1$.  The adiabatic curve (upper line) is obtained by raising the amplifying factor (corresponding to the total duration) $\tau$ of the original Hamiltonian, while the E--CD curve (lower line) is obtained by keeping the initial value of $\tau$ fixed, while raising the frequency $\omega$. The quantity on the abscissa is dimensionless, having set $\hbar=1$.}
\label{fig:intnorm}
\end{figure}
The robustness of the method against static errors in the parameters, namely amplitude and relative phase of the fields, is discussed in Appendix \ref{appendix:robustness}. The results show that the method is not very sensitive to errors: it behaves linearly with respect to amplitude noise and quadratically with respect to phase noise.

\subsection{Two-qubit entanglement creation} \label{sec:2qub}
\begin{figure}
\includegraphics[width=\linewidth]{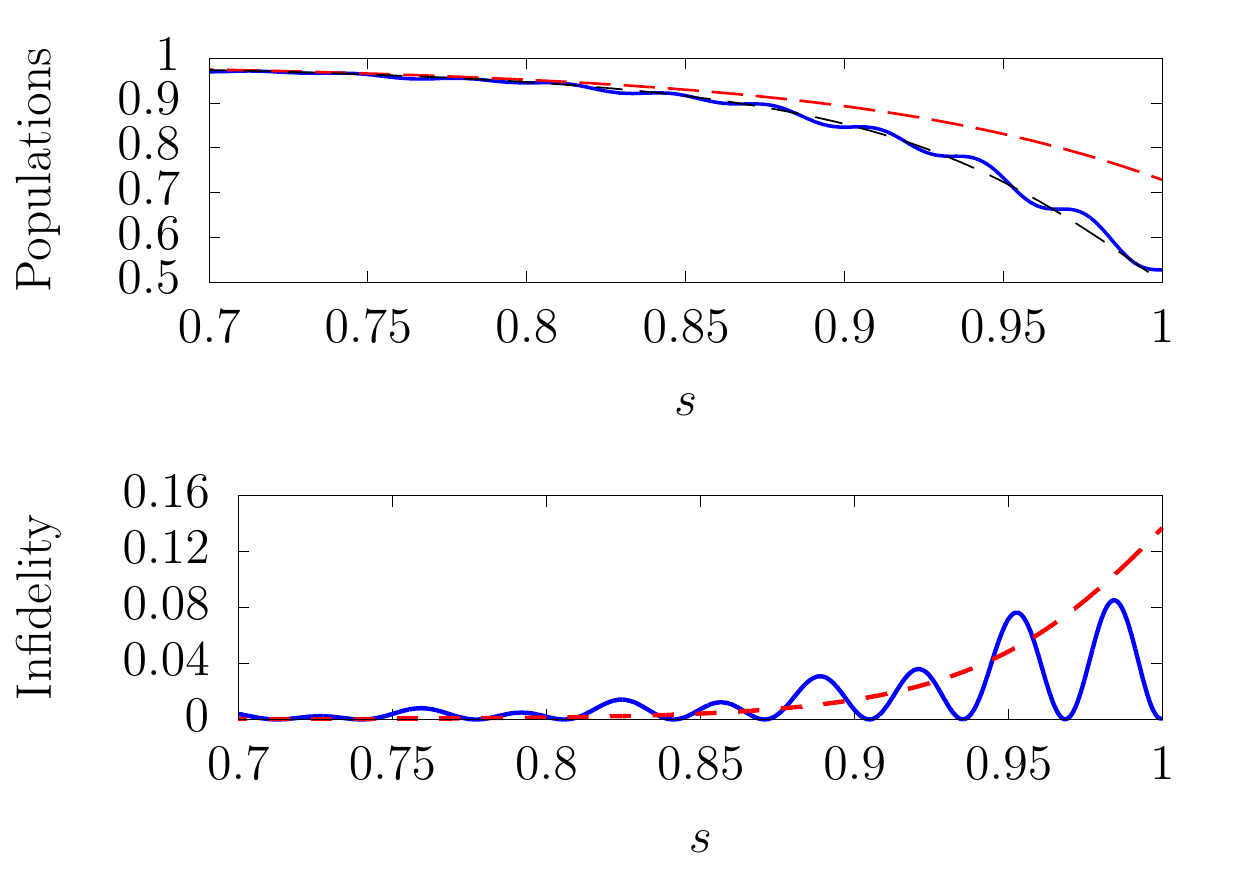}
\caption{Preparation of a two-qubit entangled state via the E--CD method (solid blue), compared against standard adiabatic (dashed red). The results are zoomed in the rescaled--time interval $0.7\le s\le 1$ for better visibility and the parameters are $\varepsilon = 5$, $\tau = 5$, $T=1/2$. Above, the evolution of the population of the bare level $\ket{00}$ is shown, with the black dashed line indicating the target dynamics. The two eigenstates $(\ket{01}\pm\ket{10})/\sqrt{2}$ are never populated, while the state $\ket{00}$ is adiabatically converted into the entangled superposition $(\ket{00}+\ket{11})/\sqrt{2}$. Below, the infidelity is represented, reaching a final value of $\sim 0.1$ for adiabatic driving and $1.7 \times 10^{-3}$ for E--CD driving. The total speed up is of $\sim 12$ times using an E--CD field $\ha$ with maximal strength satisfying $\mathcal S(\ha) = \frac{1}{2} \mathcal S(H)$ [see Eq. \eqref{eq:rationorms}].}
\label{fig:2qub}
\end{figure}
The ability to reliably produce and manipulate entanglement is at the core of quantum information processing \cite{NielsenChuang}. High-fidelity entanglement preparing protocols are thus a central touchstone for the development of quantum technologies. These are typically difficult to realize, especially due to the necessity of strong interactions and the ability of controlling many of them. 

Here, we use the E--CD method for the high-fidelity adiabatic preparation of a two-qubit entangled Bell state, requiring time-dependent control of local terms and one two-body interaction. 
Specifically, let us consider the Hamiltonian
\begin{align}
H(s) & = B(s) H_1 + g(s) H_2 \label{eq:2qubham},\\
& = - B(s) (\sz^{(1)}+\sz^{(2)}) - g(s) (\sx^{(1)} \sx^{(2)} + \sz^{(1)} \sz^{(2)}) \nonumber.
\end{align}
The Hamiltonian $H_1$ is local, while $H_2$ has entangled eigenstates. Varying the field $B(s)$ from high values to very small values, while keeping $g$ fixed, the adiabatic path of the ground state connects a separable state $\ket{00}$ to an entangled Bell state $\ket{\psi_+}=\frac{1}{\sqrt{2}}\left(\ket{00}+\ket{11}\right)$. In particular, we fix $g=1$, choose the simple time dependence $B(s) = \varepsilon(1-s)$, and consider a total duration $\tau$ of the protocol.

The dynamical Lie algebra $\mathcal L$ of the system has dimension four and is isomorphic to $\mathfrak{u}(2)$, so the system is not fully controllable, since $\mathcal L \subset \su(4)$. However, from Theorem \ref{theorem1}, we already know that the Hamiltonian $\hb$ will be inside $\mathcal L$, and will thus be a combination of the four basis elements. Since we know from Sec. \ref{sec:cddriving} that $\hb$ is orthogonal, with respect to the Hilbert-Schmidt inner product, both to $H(t)$ and its time derivative, we also know that it will not have components along $H_1$ and $H_2$. Indeed, $\hb$ can be computed analytically (Appendix \ref{appendix:diag}) and turns out to have the form
\begin{equation}
\tau \hb(s) = \frac{1}{2} \frac{\varepsilon}{4 \varepsilon^2 \left(1-s\right)^2+1} H_3 = \tau \fb(s) H_3,
\end{equation}
where $H_3= \sx^{(1)}\sy^{(2)}+\sy^{(1)}\sx^{(2)}$ is a hopping term. Therefore, $\hb$ would require the time-dependent control of an extra two-body interaction. 

In order to avoid the additional implementation of $H_3$, we can use the results of Sec. \ref{sec:su2}, since $H_3 \propto [H_1,H_2]$. We thus apply an E--CD field of the form $\ha(s) = c_1(s) H_1 + c_2(s) H_2$, with control functions
\begin{align*}
c_1 (s) = -\sqrt{  \lvert \fb(s)\rvert \omega} \cos(\omega s \tau),\\
c_2(s) = \sqrt{ \lvert \fb(s)\rvert \omega} \sin(\omega s \tau).
\end{align*}
The outcomes of a numerical simulation with parameters $\varepsilon =5$, $\tau = 5$, $N_T = \omega \tau/2 \pi = 10$ are reported in Fig. \ref{fig:2qub}. A final infidelity of $1.7 \times 10^{-3}$ is produced with a speed up of $\sim 12$ times, for $\ha$ having maximal strength $\mathcal S(\ha) \le \frac{1}{2} \mathcal S(H)$ [see Eq. \eqref{eq:strength}]. These results are extremely favorable, showing that the E--CD method provides a concrete advantage as a quantum control tool for a difficult and important task such as entanglement creation.
\subsection{Three-level system}
We apply now the E--CD scheme to the case of a three-level system undergoing a sequence of LZM avoided crossing. 
The calculation and application of the exact CD field was discussed in detail in Ref. \cite{Theisen2017}. The Hamiltonian, in terms of dimensionless quantities, reads
\begin{equation} \label{eq:3levham}
H(s) = \begin{pmatrix}
d + \varepsilon \left(s -\frac{1}{2}\right) & 1 & 0 \\
1 & -2 d & 1 \\
0 & 1 & d - \varepsilon \left(s -\frac{1}{2} \right)
\end{pmatrix}.
\end{equation}
\begin{figure}
\includegraphics[width=\linewidth]{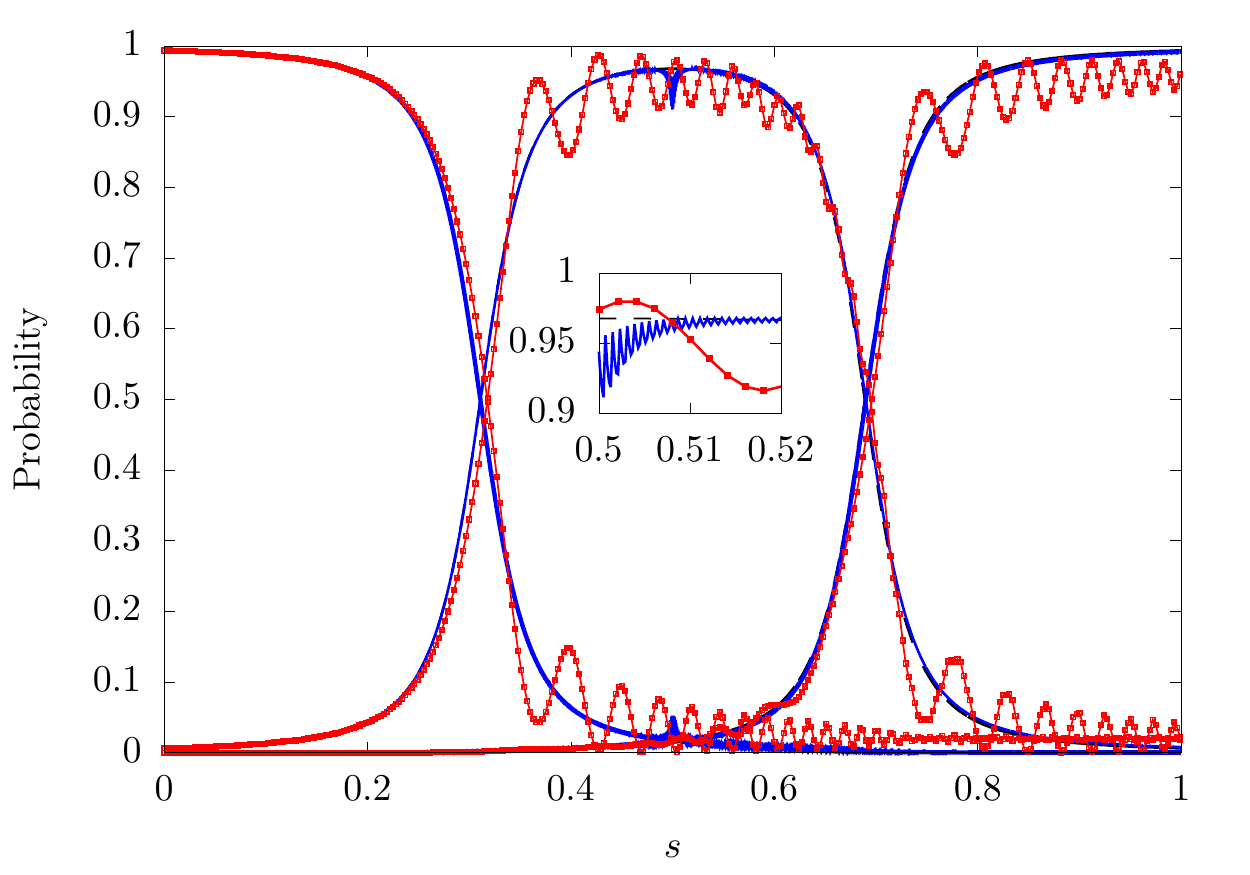}
\caption{Comparison of E--CD and adiabatic methods for the three-level system with Hamiltonian of Eq. \eqref{eq:3levham} and parameters $\tau=25$, $\varepsilon = 40$, $d = 5/2$. The evolution of the populations of the bare states is shown in dotted-solid red for adiabatic driving, by the blue solid line for E--CD driving, in dashed black for true (target) adiabatic dynamics. The inset highlights the oscillation of the E--CD dynamics around the target one.}
\label{fig:3lev}
\end{figure}
This Hamiltonian can physically describe an effective spin-1. For example, it could describe a local part of the energy spectrum of a molecular nanomagnet \cite{collison1}, which is subject to a constant magnetic field along the $x$ direction, a linearly sweeping magnetic field along the $z$ direction, and an axial zero-field splitting \cite{Theisen2017}.
The Hamiltonian is real, so we are in the situation of Theorem \ref{theorem2}.
The CD field $\hb$ cannot be computed analytically for all times, but has the general form 
\begin{equation}
\hb = \begin{pmatrix}
0 & -i \fb^{(12)} & -i\fb^{(13)} \\
i \fb^{(12)} & 0 & -i \fb^{(23)} \\
i \fb^{(13)} & i \fb^{(23)} & 0 
\end{pmatrix},
\end{equation}
with $\fb^{(12)}$ and $\fb^{(23)}$ always negative, while $\fb^{(13)}$ assumes also positive values.
The effective field is chosen of the form 
\begin{equation}
\ha = \sqrt{\omega} \begin{pmatrix}
c_3(t) & c_1(t) & 0 \\
c_1(t) & 0 & c_2(t) \\
0 & c_2(t) & -c_3(t)
\end{pmatrix}.
\end{equation}

The control functions $\bm{c}(t)$ are chosen such that the constraint equations given by $\mb^{(1)} = \ma^{(2)}$ can be easily solved analytically, ensuring precision in infidelity at least to order $T^{3}$ at the end of one period. This is done by recalling, as observed in Sec. \ref{sec:su2}, that only sines and cosines of the same harmonics contribute to the first set of constraint equations. They read
\begin{align*}
c_1(t)  = & A \cos( \omega t) - B \cos(2 \omega t) ,\\
c_2(t) = & C \sin( \omega t) - D \cos(3 \omega t) ,\\
c_3(t)  =& B\sin(2 \omega t) + D \sin(3 \omega t).\\
\end{align*} 
The constraint equations are
\begin{subequations} \label{eq:subeq}
\begin{equation}
\frac{1}{4} T B^2  = - \int^T \fb^{(12)}(t) dt ,\end{equation}
\begin{equation}
\frac{1}{6} T D^2  = - \int^T \fb^{(23)}(t) dt ,
\end{equation}
\begin{equation}
\frac{1}{2} T A C = -\int^T \fb^{(13)}(t) dt .
\end{equation}
\end{subequations}
Approximating the integrals with the value of the integrated functions multiplied by $T$, as explained in detail in Sec. \ref{sec:su2}, an interpolated solution for Eqs.
\eqref{eq:subeq} is given by
\begin{align}
A(t) &= -\sign [\fb^{(13)}] \sqrt{2 \left\lvert \fb^{(13)}\right\rvert},  \quad B(t) = 2 \sqrt{\left\lvert\fb^{(12)}\right\rvert}, \nonumber \\
C(t) &= \sqrt{2\left\lvert \fb^{(13)} \right\rvert}, \quad D(t) = \sqrt{6 \left\lvert\fb^{(23)}\right\rvert}.
\end{align}
The results with parameters $\tau=25$, $\varepsilon=40$, $d = 5/2$ are shown in Fig. \ref{fig:3lev}. A speed up of 2.5 times is obtained, for equal strength $\mathcal S(\ha) = \mathcal S(H)$ of the Hamiltonians. This shows that the E--CD paradigm behaves well also in a situation where more complex nonadiabatic phenomena are present in the dynamics of the quantum system.

\section{Conclusion}\label{sec:conclusion}

We have introduced and discussed a method for speeding-up adiabatic quantum state transfer without requiring the introduction of new Hamiltonian components, while remaining always close to the true adiabatic path. 
This is achieved by introducing a control Hamiltonian which simulates in real time the effect of a counterdiabatic field \cite{rice1,berry1}, but can be adsorbed into the initial Hamiltonian $H(t)$.
As a trade-off, complete control of $H(t)$ is required. However, the algebraic results of Sec. \ref{sec:control} can be exploited to limit the number of independent control Hamiltonians needed in $H(t)$ to implement the control protocol. For instance, we have shown in Sec. \ref{sec:examples} that, by encoding a two-level problem into a two-qubit framework, a two-qubit entangled state can be prepared with high fidelity in short time, requiring control of a single two-body interaction and a local field. 

Effective counterdiabatic Hamiltonians have been computed for the prototype applications studied, by choosing control functions which produce simple constraint equations for the amplitudes; see Secs. \ref{sec:su2} and \ref{sec:examples}. Remarkably, this allowed us to solve the equations analytically, and to obtain satisfying results without resorting to heavy numerical methods. Nonetheless, the structure of the E--CD problem still leaves room for optimization of the parameters, so further improvement can be predicted by the hybridization with optimal control strategies \cite{Glaser2015}. Moreover, 
one could combine the E--CD method with a preselection of an optimal adiabatic protocol. For instance, one might first conceive a ``local'' adiabatic driving \cite{Roland2002}, in which the evolution is accelerated where nonadiabatic effects are small and then slowed down in the vicinity of avoided crossings. This would provide an \emph{a priori} improvement of the basic adiabatic protocol with respect, for example, to a simple LZM sweep. Such a protocol could then be ulteriorly sped up by means of E--CD corrections.

In particular, the simplicity and versatility of the E--CD method makes it a strong strategy for the control 
of adiabatic processes. Therefore, we expect it to be of practical interest for many applications, ranging from quantum state preparation to adiabatic quantum computing \cite{fahri2000,albash2016}, quantum thermodynamics \cite{Alicki2018,Campo2014}, and probing of quantum critical dynamics \cite{Greiner2002,Polkovnikov2005,delcampo2012,manybody2014}.

The method has been developed for finite-dimensional quantum systems. Still, the scheme could in principle be applied to infinite-dimensional systems with discrete levels, or restricted to well-separated nonpathological parts of the spectrum. Pathological here refers to possible degeneracies or continuous spectra. The first case, which is problematic also for finite dimensions, would require addressing the decoupling of instantaneous eigensubspaces, rather then single instantaneous eigenvectors. On the other hand, continuous spectra would require a revision of the CD framework, starting from the adiabatic theorem itself  \cite{Maamache2008}. Concerning the more mathematical results, noncompactness of the dynamical Lie group of the system would break some of the assumptions used in our control-theoretic setting \cite{Genoni2012}, which should then be investigated more thoroughly.

\acknowledgments

F.P. is grateful for hospitality at Imperial College London and acknowledges partial funding from the Doctoral International Mobility Program of the University of Parma. B.D. acknowledges funding from the Engineering and Physical Sciences Research Council (UK) administered by Imperial College London via the Postdoctoral Prize Fellowship program. The authors are grateful to Riccardo Mannella for a critical reading of the manuscript.

\appendix

\section{Proof of Theorem \ref{theorem1}} \label{appendix:proofs}

Before proceeding with the proof of Theorem \ref{theorem1}, and to fix notation, let us recall the basic formulation of the adiabatic theorem.

Considering the Schr\"odinger equation, $ i \hbar \partial_t U(t)=H(t) U(t)$, from initial time $t_i$, let us parametrize the physical time $t$ according to $t-t_i=s \tau$, $ 0\le s \le 1$, where $\tau=t_f - t_i$ is the total evolution time. The equation becomes 
\begin{equation}\label{adeq}
i \hbar \partial_s U_\tau(s)=\tau H (s) U_\tau(s) .
 \end{equation}
\begin{thm}{\normalfont(Adiabatic theorem \cite{messiah1961qm,kato1950})}
Let $U_\tau(s)$ be the solution of \eqref{adeq}. Then,
\begin{equation} \label{adtheo1}
U_\tau(s)  - \sum_n e^{-\frac{i}{\hbar} \tau\int_0^s E_n(t(s')) ds' }\ket{n(t(s))}\bra{n(0)}  = O(\tau^{-1}),
\end{equation}
\end{thm}
In the adiabatic limit $\tau\to \infty$, and for $s=1$, it holds exactly
\begin{equation}\label{eq:adapprox}U_\tau(1) \overset{\tau\to \infty} {\longrightarrow} U_\text{ad}(t_i,t) \equiv   \sum_n e^{-\frac{i}{\hbar}\int_{t_i}^t E_n(t') dt' }\ket{n(t)}\bra{n(0)}.
\end{equation}

Under the control-theoretic setup introduced in Sec. \ref{sec:control}, we proceed with the proof of Theorem \ref{theorem1}. 

\begin{proof} \normalfont
Let $\hat{H}_\text{CD} = \hb(\,\hat{t} \,)$ be the time-independent matrix obtained by evaluating $\hb(t)$ at the time instant $t=\hat{t}$.
In order to prove that $\hb \in \mathcal L$, one can equivalently show that either {(i)} $\hb$ can be written at all times as a linear combination of elements of a basis of $\mathcal L$, or {(ii)} that the group elements generated by the matrices $\hat{H}_\text{CD}$ for all time instants $\hat t$, $e^{-i \hat{H}_\text{CD} t}$, can be written as a concatenation of elements of the group $e^{\mathcal L}$. We follow route {(ii)}. 
With the initial time $t_i$ set to zero and $\hbar$ set to one for simplicity of notation, let $\ub(t)$ be the solution of the Schr\"odinger equation 
$$i  \partial_t \ub(t) = \hb(t) \ub(t) = \big[\hcd(t)-H(t)\big] \ub(t),$$
where the Hamiltonian $H_\text{corr}$ was introduced in Sec. \ref{sec:cddriving}. Let us assume that a Magnus expansion (see Appendix \ref{appendix:Magnus}) $\ub(t)=e^{\mb(t)}$ can be formally written, where
\begin{align}
\mb(t) & = -i \int_0^t (\hcd(t_1)-H(t_1)) d t_1 \nonumber\\
&-\frac{1}{2} \int_0^t dt_1 \int_0^{t_1} d t_2 [\hb(t_1), \hb(t_2)]+ \dots \label{eq:MB}
\end{align}
Now, calling $\mc(t)$ and $M(t)$ the Magnus exponents generated by $\Hcd$ and $H$ respectively, the terms in \eqref{eq:MB} can be rearranged as to give
\begin{align}
\mb(t) & = \mc(t) - M(t) + R(t),
\end{align}
with
\begin{align*}
R(t) = & -\frac{1}{2} \int_0^t dt_1 \int_0^{t_1} dt_2 [\Hcd(t_1), H(t_2)] \\
&  -\frac{1}{2} \int_0^t dt_1 \int_0^{t_1} dt_2 [H(t_1), \Hcd(t_2)] + \dots
\end{align*}
By expanding in Taylor series around the midpoint $t^*=t/2$ of the (small) integration interval $(0,t)$, see Eq. \eqref{eq:magnusmix} in Appendix \ref{appendix:Magnus}, one can write 
$$R(t) = -\frac{1}{12} \left( [\hcd^{(1)},H^{(0)}]-[\hcd^{(0)},H^{(1)}] \right) t^3 + o(t^4),$$
where we have used the notation 
$$H_X^{(k)} = \frac{1}{k!} \left.\frac{\partial^{k} H_\text{X}(t)}{\partial^k t}\right\lvert_{t^*}.$$
Therefore $R(t)=o(t^3)$.

Now, considering small times $t$, and using repeatedly the Zassenhaus formula \cite{Magnus1954} to factorize the exponential $\exp \{ \mb(t) \}$, $\ub$ can be decomposed like
$$\ub(t) = \ucd(t) U^{\dagger}(t) e^{-\frac{1}{2}[M_\text{CD}(t),M(t)]} e^{o(t^3)}. $$
The unitary matrix $\ucd(t)$, introduced in Sec. \ref{sec:cddriving}, is the solution of $i \partial_t \ucd = H_\text{corr} \ucd$. 
Proceeding as above, one can estimate $-\frac{1}{2}[M_\text{CD}(t),M(t)]=o(t^2)$.
The adiabatic theorem, Eqs. \eqref{adtheo1} and \eqref{eq:adapprox}, then states that $\ucd(t)$ can be approximated to arbitrary precision through an adiabatic process of long duration $\tau$,
\begin{equation} \label{eq:UBapprox}
\ub(t) =\uad(t) U^{\dagger}(t) e^{o(t^2)}  + o(\tau^{-1}). \end{equation}
We are then ready to prove the theorem. First of all, let us observe that, for sufficiently small $t$, one can rewrite
\begin{equation}\label{eq:discret}
\hat{H}_\text{CD} = \frac{\partial}{\partial {\hat t}} \int_0^{\hat t} \hb(t_1) dt_1 = \frac{1}{t} \int_{\hat t}^{\hat{t}+t} \hb(t_1) dt_1 + o(t).
\end{equation}
Applying the exponential map to Eq. \eqref{eq:discret} and using Eq. \eqref{eq:UBapprox} we eventually obtain
\begin{align}
e^{-i \hat{H}_\text{CD} t} & = \ub(\hat{t},\hat{t}+t)e^{o(t^2)} \nonumber, \\
& = U_\text{ad}(\hat{t},\hat{t}+t)U^{\dagger} (\hat{t},\hat{t}+t)e^{o(t^2)}+o(\tau^{-1}). \label{eq:proof1}
\end{align}
Since $U_\text{ad}$ and $U$ can be realized via time dependent control of $H(t)$, they belong to the group $e^{\mathcal L}$.
Equation \eqref{eq:proof1} then means that, for all times $\hat{t}$, the group element $\exp(-i \hat{H}_\text{CD} t)$ generated by $-i \hat H_\text{CD}$ can be approximated by a concatenation of elements of the group $e^{\mathcal L}$, arbitrarily well for sufficiently small $t$ and large $\tau$. That is, $\hat{H}_\text{CD} \in \mathcal L$ for all $\hat{t}$.
\end{proof}
\section{Magnus expansion} \label{appendix:Magnus}
\subsection{Basics and useful bounds}

The Magnus expansion \cite{Magnus1954,Blanes2009} is a representation of the evolution operator $U(t)$, solution of the time-dependent Schr\"odinger equation $i \hbar \partial_t U(t) = H(t) U(t)$, in exponential form $U(t) = \exp\{M(t)\}$. The exponent $M(t)$ is an infinite sum $M(t) = M^{(1)}(t)+M^{(2)}(t) + \dots $ whose first terms read
\begin{align}
 M^{(1)}(t) =&  -\frac{i}{\hbar} \int_0^t H(t_1) dt_1, \\
 M^{(2)}(t) =& \left(\frac{-i}{\hbar}\right)^2 \int_0^t d t_1\int_0^{t_1} d t_2  [H(t_1),H(t_2)], \\
 M^{(3)}(t)  = &\left(\frac{-i}{\hbar}\right)^3 \int_0^t d t_1\int_0^{t_1} d t_2 \nonumber \\
 &\int_0^{t_3} dt_3 \Big( [H(t_1),[H(t_2),H(t_3)]] \nonumber\\
 & +[H(t_3),[H(t_2),H(t_1)]]\Big).
\end{align}

We will in general denote with $M_\text{X}^{(k)}$ the $k$-th Magnus term for the Hamiltonian $H_\text{X}$.
The behavior of the Magnus terms for small integration time $t$ can be conveniently studied by Taylor-expanding the Hamiltonian around the midpoint of the integration interval, and computing the integrals afterwards \cite{Blanes2009}. It turns out that, due to the time symmetry of the expansion, all terms are odd functions of $t$. In particular, one has
\begin{equation} \label{eq:magorder}
M^{(2 m)}(t) = o(t^{2m+1}); \quad M^{(2m+1)}=o(t^{2m+3}).
\end{equation}
Due to its usefulness in Appendix \ref{appendix:proofs}, for the proof of Theorem \ref{theorem1}, we do the explicit calculation for the following ``generalized'' Magnus term, which involves the commutator of two different matrices,
\begin{equation}
\Omega_2(t)  = \frac{1}{2} \int_0^t dt_1 \int_0^{t_1} d t_2 [A(t_1),B(t_2)]. \label{eq:gen2mag}
\end{equation}
Let us denote with $t^*=t/2$ the midpoint, and let us introduce the notation 
$$a_k = \frac{1}{k!} \left.\frac{\partial^{k} A(x)}{\partial^k x}\right\lvert_{t^*}, b_k = \frac{1}{k!} \left.\frac{\partial^{k} B(x)}{\partial^k x}\right\lvert_{t^*},$$
Inserting the Taylor expansions $A(t) = \sum_k a_k (t-t^*)^k$ and $B(t) = \sum_k b_k (t-t^*)^k$ into Eq. \eqref{eq:gen2mag} one obtains
\begin{align}
\Omega_2(t) & = \frac{1}{2} \sum_{k,n}^{0,\infty} \frac{1}{n+1}  \Big\{\frac{1-(-1)^{k+n+2}}{k+n+2} \nonumber \\
& - \frac{(-1)^{n+1}[1-(-1)^{k+1}]}{k+1} \Big\} [a_k,b_k] \left(\frac{t}{2}\right)^{n+k+2} \nonumber, \\
& = \frac{1}{4} [a_0,b_0] t^2 + \frac{1}{24}\left( [a_1,b_0]-[a_0,b_1] \right) t^3 \nonumber \\
& + \frac{1}{48}\left([a_0,b_2]+[a_2,b_0] \right) t^ 4 + o(t^5). \label{eq:magnusmix}
\end{align}

\subsection{Error in the infidelity for the E--CD method} \label{sec:infiderror}

It is useful in our study to estimate the error when the expansion is truncated. In particular, we will focus on the error in the infidelity at the end of one period $T$ of the oscillating control functions, for the problem of approximating $\ub$ by means of $\ua$, and for the class of Hamiltonians characterized by Theorem \ref{theorem2} (see the end of Sec. \ref{sec:approximate}). We will assume that the system starts in the ground state $\ket{gs(t_i)}$ at the initial time $t_i$. Writing the states in terms of the evolution operators, and denoting with $\langle \cdot \rangle$ the expectation value $\bra{gs(t_i)} \cdot \ket{gs(t_i)}$, the infidelity can be written in the form 
\begin{equation}\label{eq:infidprop}\If = 1- |\langle \ub^{\dagger} \ua \rangle |^2.\end{equation} 
We then write the propagators in terms of their Magnus expansions, $\ub^{\dagger}(T)=e^{-\mb(T)}$ and $\ua(T) = e^{\ma(T)}$, and use the Baker-Campbell-Hausdorff formula for computing the product $\ub^{\dagger} \ua$. Assuming that one has solved the constraint equations, discussed in Sec. \ref{sec:approximate}, involving the first $2m$ Magnus terms of $\ua$, then it holds 
$$\sum_{k=1}^m \mb^{(k)} = \sum_{k=1}^{2m}  \ma^{(k)}.$$
We can then write
\begin{align}
\langle \ub^{\dagger } \ua \rangle =& \left\langle \exp \left\{\ma^{(2m+1)} + \ma^{(2m+2)}-\mb^{(m+1)} \right. \right. \nonumber \\
& -\frac{1}{2} \left[\mb^{(1)},\ma^{(2m+1)} \right] \\
 & \left. \left. -\frac{1}{2} \left[\mb^{(m+1)},\ma^{(2)} \right] + \dots\right\}\right\rangle .\label{eq:expvalue}
\end{align}
Let us remark that, since the oscillations in $\ha$ are much faster than the timescales of $\hb$, the behavior in Eq. \eqref{eq:magorder} does hold for $\mb(T)$, that is
$\mb^{(2 m)}(T) = o(T^{2m+1})$ and $\mb^{(2m+1)}(T)= o(T^{2m+3})$.
For the dynamics $\ua$, though, the integration is over a full period and the same is thus not true. Recalling the prefactor $\sqrt{\omega}$ in $\ha$, the E--CD Magnus terms are of order $\ma^{(m)}(T) = o(T^{m/2})$.  Equation \eqref{eq:expvalue} can then be expressed like
\begin{equation}\label{eq:orderexp}
\langle \ub^{\dagger } \ua \rangle = \left\langle \exp \left\{\ma^{(2m+1)} + o\left(T^{m+1}\right)\right\}\right\rangle. \end{equation}
One can now expand the matrix exponential in Taylor series and take the expectation value termwise. Taking the modulus squared, for a generic Hermitian matrix $X$,  it holds
\begin{equation}\label{eq:expformula}\lvert \langle e^{i X t} \rangle \rvert^2 = 1 +  t^2 \Big( \text{Re} \langle X^2 \rangle + |\langle X \rangle|^2 \Big) + o(t^3), \end{equation}
Applying the same reasoning to Eq. \eqref{eq:orderexp} and inserting into Eq. \eqref{eq:infidprop} one finally obtains
\begin{align}
\If(T) & = - \text{Re}\left\langle \left( \ma^{(2m+1)} \right)^2 \right\rangle - \left\lvert \left\langle \ma^{(2m+1)} \right\rangle  \right\rvert^2 + o(T^{2m+\frac{3}{2}}), \label{eq:infidorder} \\
& = o(T^{2m+1}).
\end{align}

\section{Robustness}\label{appendix:robustness}

We test the robustness of the E--CD method against possible experimental imperfection for the LZM case discussed in Sec. \ref{sec:LZM}. We take into account possible errors in the amplitude and relative phase of the driving fields. This is done by adding a static relative offset $\delta$. More specifically, with reference to Eq. \eqref{eq:contham},
\begin{align}
\sin(\omega s \tau) & \to (1+\delta) \sin(\omega s \tau), \nonumber \\
\sin(\omega s \tau) & \to \sin(\omega s \tau + 2 \pi \delta), \label{eq:robustness}
\end{align}
with $-1/2 \le \delta \le 1/2$. The quantity $F = |\braket{gs(t_f) \lvert \psi(t_f)}|^2$ is the final fidelity between system state $\ket{\psi(t_f)}$ and target (ground) state $\ket{gs(t_f)}$, while $F_0$ is the fidelity in the absence of errors. The behavior of the relative error $|1-F/F_0|$ is shown in Fig. \ref{fig:stability} as a function of the offset $\delta$ for the two error situations described by Eq. \eqref{eq:robustness}. Both positive and negative values of $\delta$ are considered. The E--CD method results in being mostly linearly stable against amplitude errors, while it turns out to be quadratically sensitive to phase errors. It should be remarked that in some cases small errors in amplitude can incidentally lead to a small improvement of the method: this can be seen from the two regimes observable for positive $\delta$. The negative peak around $10^{-1}$ indicates the transition between a region $\delta<10^{-1}$ where the protocol improves and a region for greater values where it performs worse instead. The behavior is symmetric with respect to the sign of $\delta$ for phase noise.

\begin{figure}[t!]
\includegraphics[width=\linewidth]{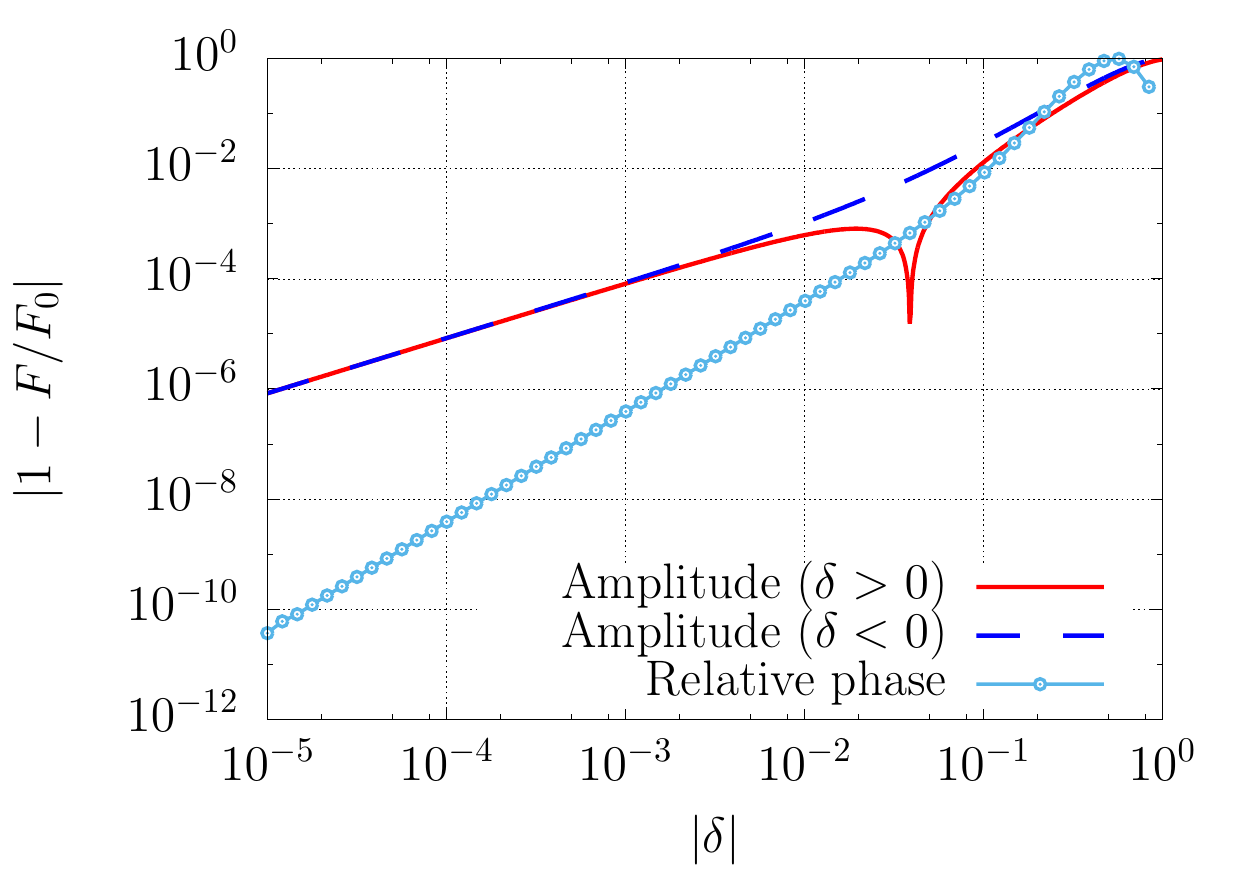}
\caption{Robustness of the E--CD method, Eq. \eqref{eq:contham}, against imperfections in amplitude and relative phase of the control fields as considered in Eq. \eqref{eq:robustness}. For the amplitude, the solid line indicates the results for relative error $\delta>0$, while the dashed line is used for $\delta <0$. In the case of the relative phase, the results are completely symmetric with respect to the sign of $\delta$ (the two lines are overlapping). The dependence on $\delta$ is linear for the amplitude for $|\delta| < 10^{-2}$. It is quadratic instead for the relative phase.}
\label{fig:stability}
\end{figure}

\vspace{0.5cm}

\section{CD field for the two-qubit problem} \label{appendix:diag}
In this appendix, we report the explicit computation of the CD field for the two-qubit problem described in Sec. \ref{sec:2qub}. The initial Hamiltonian in dimensionless units, see Eq. \eqref{eq:2qubham}, is
$$H(s)/\hbar = -\varepsilon(1-s) [\sz^{(1)}+\sz^{(2)}]-[\sx^{(1)} \sx^{(2)}+\sz^{(1)}\sz^{(2)}]. $$
As a first step, we need to diagonalize $H(s)$. First of all, $H(s)$ can be partially diagonalized by writing it in the (time-independent) basis
$$\mathcal B_1 = \left\{  \ket{00}; \frac{\ket{01}+\ket{10}}{\sqrt{2}}; \frac{\ket{10}-\ket{01}}{\sqrt{2}}; \ket{11} \right\}.$$
Let $Q$ be the change-of-basis matrix from basis $\mathcal B_0 =\{\ket{00}, \ket{01},\ket{10},\ket{11} \}$ to $\mathcal B_1$, thus having the kets in $\mathcal B_1$ as columns. 
The Hamiltonian reads
\begin{equation}\label{eq:ham1}
\begin{pmatrix}
-1+2 \varepsilon(s-1) & 0 & 0 & -1 \\ 
 & 0 & 0 & 0 \\
0 & 0 & 2 & 0 \\
-1& 0 & 0 & -1-2\varepsilon(s-1)
\end{pmatrix}.
\end{equation}
We thus see that two levels are actually decoupled among themselves and from the rest of the spectrum.
What remains to be diagonalized is a two-by-two real symmetric matrix (formed by the four corner elements). For this we can use the usual convenient trigonometric parametrization of two-by-two unitary matrices \cite{rice1}, and the matrix diagonalizing \eqref{eq:ham1} is
\begin{equation*}
P(s)=\begin{pmatrix}
-\sin[\theta(s)] & 0 & 0 & \cos[\theta(s)]\\
0 & 1 & 0 & 0 \\
0 & 0 & 1 & 0 \\
\cos[\theta(s)] & 0 & 0 & \sin[\theta(s)]
\end{pmatrix},
\end{equation*}
with $\theta(s) = \frac{1}{2} \arctan \left[\frac{1}{2 \varepsilon (1-s)} \right]. $
Eventually, the full diagonalizing matrix is $U(s)=P(s) Q$, and thus  the CD field can be readily computed, recalling that $s=(t-t_i)/\tau$,
\begin{align*}
\hb(s)&  = \frac{i}{\tau} \frac{\partial U(s)}{\partial s} U(s)^{\dagger}, \\
& = \frac{1}{2 \tau} \frac{\varepsilon}{1+4 \varepsilon^2 (s-1)^2} [\sx^{(1)}\sy^{(2)}+\sy^{(1)} \sx^{(2)}].
\end{align*}

\end{document}